\documentclass[manuscript,nonacm,acmsmall]{acmart}

\setcitestyle{numbers}
\AtEndPreamble{\theoremstyle{acmplain}
    \newtheorem{obs}[theorem]{Observation}
\theoremstyle{acmdefinition}
    
    \newtheorem{remark}[theorem]{Remark}}

\usepackage{mathtools}
\usepackage{faktor}
\usepackage{mathrsfs} \usepackage{tikz}
\usetikzlibrary{arrows,shapes,cd,decorations.markings}

\DeclareMathAlphabet{\mathbbm}{U}{bbold}{m}{n}

\usepackage{enumitem}

\usepackage{xcolor}
\definecolor{blue}{HTML}{1E88E5}
\definecolor{yellow}{HTML}{FFC107}

\makeatletter

\DeclarePairedDelimiter\abs{\lvert}{\rvert}

\let\geom\abs
\let\subset\subseteq
\DeclareMathOperator{\height}{ht}
\DeclareMathOperator{\ar}{ar}
\DeclareMathOperator{\CSP}{CSP}
\DeclareMathOperator{\PCSP}{PCSP}
\DeclareMathOperator{\pol}{pol}
\DeclareMathOperator{\spol}{spol}

\DeclareMathOperator{\Hom}{Hom}
\DeclareMathOperator{\mhom}{mhom}
\DeclareMathOperator{\hp@l}{hpol}
\newcommand{\hpol}[2][\relax]{\ifx #1\relax \hp@l(#2) \else \hp@l^{(#1)}(#2) \fi}
\DeclareMathOperator{\Ann}{Ann}

\DeclareMathOperator{\coker}{coker}
\DeclareMathOperator{\im}{im}

\DeclareMathOperator{\@d}{id}
\newcommand{\id}{{\@d}}

\let\bd\partial
\let\rel\mathbf
\let\minion\mathscr

\let\frk\mathfrak
\newcommand{\ssX}{{X}}
\newcommand{\ssY}{{Y}}

\newcommand{\eml}{Y}
\newcommand{\maptosigma}{t}

\newcommand{\NP}{\textsf{NP}}
\newcommand{\Ptime}{\textsf{P}}

 \let\R\reals
\newcommand\ints{{\mathbb Z}}  \let\Z\ints
\newcommand{\codd}{C_\ell}
\newcommand{\Ztwo}{{\ints_2}}
\newcommand{\zplus}{\mathfrak{Z}_+}
\newcommand{\zminus}{\mathfrak{Z}_-}
\newcommand{\affine}{\minion Z}

\def\imod#1{\allowbreak\mkern10mu({\operator@font mod}\,\,#1)}

\newcommand{\slice}{s}
\DeclareMathOperator{\diag}{diag}

\newcommand{\blueb}{{\color{blue} \bullet}}
\newcommand{\yellowb}{{\color{yellow} \bullet}}
\newcommand{\grLambda}{{{\mathbb Z}[{\mathbb Z}_2]}}
\newcommand{\gromega}{{\nu}}

\let\epsilon\varepsilon
\makeatother

\setcopyright{cc}
\setcctype{by}
\copyrightyear{2025}

\title{Hardness of 4-Colouring G-Colourable Graphs}

\author{Sergey Avvakumov}
\email{savvakumov@gmail.com}
\affiliation{\institution{Tel Aviv University}
  \city{Tel Aviv}
  \country{Israel}
}
\affiliation{\institution{Charles University}
  \city{Prague}
  \country{Czech Republic}
}

\author{Marek Filakovský}
\email{filakovsky@math.muni.cz}
\affiliation{\institution{Masaryk University}
  \city{Brno}
  \country{Czech Republic}
}

\author{Jakub Opršal}
\email{j.oprsal@bham.ac.uk}
\affiliation{\institution{University of Birmingham}
  \city{Birmingham}
  \country{UK}
}

\author{Gianluca Tasinato}
\email{gianluca.tasinato@cs.ox.ac.uk}
\affiliation{\institution{ISTA}
  \city{Klosterneuburg}
  \country{Austria}
}
\affiliation{\institution{University of Oxford}
  \city{Oxford}
  \country{UK}
}

\author{Uli Wagner}
\email{uli@ist.ac.at}
\affiliation{\institution{ISTA}
  \city{Klosterneuburg}
  \country{Austria}
}


\begin{abstract}
 We study the complexity of a class of promise graph homomorphism problems. For a fixed graph $H$, the $H$-colouring problem is to decide whether a given graph has a homomorphism to $H$. By a result of Hell and Nešetřil, this problem is NP-hard for any non-bipartite loop-less graph $H$. Brakensiek and Guruswami [SODA 2018] conjectured the hardness extends to promise graph homomorphism problems as follows: fix a pair of non-bipartite loop-less graphs $G$, $H$ such that there is a homomorphism from $G$ to $H$, it is NP-hard to distinguish between graphs that are $G$-colourable and those that are not $H$-colourable. We confirm this conjecture in the cases when both $G$ and $H$ are 4-colourable. This is a common generalisation of previous results of Khanna, Linial, and Safra [Comb. 20(3): 393-415 (2000)] and of Krokhin and Opršal [FOCS 2019]. The result is obtained by combining the algebraic approach to promise constraint satisfaction with methods of topological combinatorics and equivariant obstruction theory.
\end{abstract}


\begin{document}

\maketitle

\begin{acks}
  This research was supported by the Austrian Science Fund (FWF project P31312-N35), by project MSCAfellow5\_MUNI (CZ.\penalty20 02.\penalty20 01.\penalty20 01/\penalty20 00/22\_010/0003229) financed by the Ministry of Education, Youth and Sports of the Czech Republic and by Masaryk University GAMU MASH StG/CoG project MUNI/SC/1957/2024.
  This project has also received funding from the European Union's Horizon 2020 research and innovation programme under the Marie Skłodowska-Curie Grant Agreement No 101034413.
\end{acks}

\section{Introduction}

Deciding whether a given finite graph is $3$-colourable (or, in general, $k$-colourable, for a fixed $k\geq 3$) was one of the first problems shown to be \NP-complete (\citet{Karp1972}).
Since then, the complexity of \emph{approximating} the chromatic number of a graph has been studied extensively; in particular, it is known that 
the chromatic number of an $n$-vertex graph cannot be approximated in polynomial time within a factor of $n^{1-\varepsilon}$, for any fixed $\varepsilon>0$, unless $\Ptime=\NP$ (\citet{Zuckerman}).

However, this inapproximability result only applies to graphs whose chromatic number grows with the number of vertices; by contrast, the case of graphs with bounded chromatic number is much less well understood. For instance, given an input graph that is promised to be $3$-colourable, what is the complexity of finding a colouring of $G$ with some larger number $k>3$ of colours? Khanna, Linial, and Safra~\cite{KLS00} proved that this is \NP-hard for $k=4$ (see also \cite{GuruswamiK04,BrakensiekG16}), and only quite recently Bulín, Krokhin, and Opršal \cite{BKO19} showed \NP-hardness for $k=5$. On the other hand, the currently best polynomial-time algorithm for colouring $3$-colourable graphs, due to Kawarabayashi, Thorup, and Yoneda~\cite{KawarabayashiTY24}, uses $k=\tilde O(n^{0.19747})$ colours, where $n$ is the number of vertices of the input graph.

In general, it is believed that colouring $c$-colourable graphs with $k$ colours is \NP-hard for all constants $k \geq c \geq 3$. However, the best results known prior to our research (apart from the above) are \NP-hardness for $c=4$ and $k=7$ (Bulín, Krokhin, and Opršal \cite{BKO19}), and for $c\ge 5$ and $k=\binom c{\lfloor c/2\rfloor}-1$ (Wrochna and Živný \cite{WZ20}). Moreover, conditional hardness results --- assuming different variants of Khot's \emph{Unique Games Conjecture} --- have been obtained for all $k\geq c \geq 3$ by Dinur, Mossel, and Regev \cite{DMR09}, Guruswami and Sandeep \cite{GS20}, and Braverman, Khot, Lifshitz, and Minzer~\cite{BKLM22}.
Very recently, Fei, Minzer, and Wang~\cite{FeiMW26} announced a proof of one of these variants, namely the \emph{4-to-1 Games Conjecture} used in \cite{GS20}.

In the present paper, we study a generalisation of this question. A \emph{graph homomorphism} $f\colon G\to H$ between two graphs is a map $f\colon V(G)\to V(H)$ between the vertex sets that preserves edges, i.e., $(u,v)\in E(G)$ implies $(f(u),f(v))\in E(H)$; we write $G \to H$ if such a homomorphism exists. Throughout this paper, we assume all graphs to be finite and undirected and we treat them as symmetric binary relational structures, i.e., we view the edge set $E(G)$ as a subset of $V(G) \times V(G)$ that satisfies $(u,v)\in E(G)$ if and only if $(v,u)\in E(G)$, and we allow loops, i.e., edges of the form $(v,v)$. A graph homomorphism $f\colon G\to H$ is also called an $H$-\emph{colouring} of $G$ since a $k$-colouring of $G$ is the same as a homomorphism from $G$ to the complete (loopless) graph $K_k$ on $k$ vertices.
Vastly generalising the fact that $k$-colouring is \NP-hard if $k\geq 3$, and in \Ptime{} if $k \le 2$, Hell and Nešetřil \cite{HellN90} proved the following dichotomy: For every fixed graph $H$, deciding whether a given input graph admits an $H$-colouring is \NP-complete, unless $H$ is bipartite or has a loop, in which case the problem is in \Ptime{}. 
Analogously to approximate graph colouring, it is natural to consider the complexity of the following \emph{promise graph homomorphism problem}: Fix two graphs $G$ and $H$ such that $G \to H$. What is the complexity of $H$-colouring graphs that are promised to be $G$-colourable? More precisely, we consider the decision version of this problem, denoted by $\PCSP(G, H)$: Given an input graph $I$, output \textsf{YES} if $I\to G$ and \textsf{NO} if $I\not\to H$ (no output is required if neither is the case). Brakensiek and Guruswami conjectured \cite[Conjecture 1.2]{BrakensiekG21} that $\PCSP(G, H)$ is \NP-hard for all non-bipartite, loopless graphs $G$ and $H$ (i.e., unless the problem is guaranteed to lie in \Ptime{} by the Hell–Nešetřil dichotomy). This problem fits into the much broader framework of \emph{promise constraint satisfaction problems} (\emph{PCSP}s), from which the notation is adopted; see the survey \cite{KO22} for more background.

As a first step towards the Brakensiek--Guruswami conjecture, Krokhin and Opršal \cite{KO19} showed that $\PCSP(G,K_3)$ is \NP-hard for every $3$-colourable non-bipartite graph $G$. Their proof was based on ideas from algebraic topology; this topological intuition was formalised by Wrochna and Živný~\cite{WZ20} (and in the joint journal version \cite{KOWZ23}). We extend this to $4$-colouring:

\begin{theorem} \label{thm:main}
Let $G$ be a non-bipartite 4-colourable graph. Then $\PCSP(G, K_4)$ is \NP-hard.
\end{theorem}

The proof of Theorem~\ref{thm:main}, whose structure we present in detail in Section~\ref{sec:overview} below, builds on and significantly extends the topological approach used in \cite{KOWZ23}, bringing to bear more powerful tools from algebraic topology as well as more refined combinatorial arguments. By a simple reduction, Theorem~\ref{thm:main} reduces to the special case where $G=\codd$ is a cycle of arbitrary odd length $\ell \geq 3$ (Theorem~\ref{thm:main-cycles}). 

We rely on a general algebraic theory of PCSPs by Barto, Bulín, Krokhin, and Opršal~\cite{BBKO21}, which guarantees that the complexity of the problem only depends on its \emph{polymorphisms}. In our case, a polymorphism is a graph homomorphism $f\colon \codd^n \to K_4$ where $n$ is a natural number and $\codd^n=\codd \times \dots \times \codd$ is the $n$-fold power of $\codd$ (see Section~\ref{sec:Polymorphisms}).

The proof of Theorem~\ref{thm:main-cycles} has two parts.
In the first part of the proof, we use topological methods (\emph{homomorphism complexes} and \emph{equivariant obstruction theory}) to show that with every polymorphism $f$, we can associate a map $\phi(f)\colon \Z_2^n\to \Z_2$ of the form $\phi(f)(x_1,\dots, x_n)=\sum_{i=1}^n \alpha_i x_i$, for some 
$\alpha_i\in \Z_2$ such that 
$\sum_{i=1}^n \alpha_i$ is odd 
(Lemma~\ref{lem:minion-homomorphism}). Moreover, the map $f\mapsto \phi(f)$ preserves natural \emph{minor relations} between polymorphisms that arise from substituting and permuting variables; in technical terms, $\phi$ is a \emph{minion homomorphism} (see Definition~\ref{def:minion-homomorphism}).
The second part of the proof uses combinatorial arguments to show (Theorem~\ref{thm:bounded-arity}) that the affine maps $\phi(f)\colon \Z_2^n\to \Z_2$ arising from polymorphisms are of \emph{bounded essential arity}: the number of non-zero coefficients $\alpha_i\in \Z_2$ describing $\phi(f)$ is at most $O(\ell^2)$, independently of $n$. Hardness of $\PCSP(\codd,K_4)$ then follows from a criterion (Theorem~\ref{thm:hardness}) obtained as part of the general algebraic theory developed in~\cite{BBKO21}.

\paragraph{Related work}
Graph colouring and $H$-colouring are examples of \emph{constraint satisfaction problems} (\emph{CSP}s), a general framework that encompasses many other fundamental problems including 3SAT, HornSAT, solving systems of linear equations, and linear programming. CSPs can be formulated in several equivalent ways; the one most relevant for us is in terms of homomorphisms between relational structures: fix a relational structure $\rel A$ (e.g., a graph, a digraph, or a uniform hypergraph), the CSP with template $\rel A$, denoted by $\CSP(\rel A)$, is the problem of deciding whether a given input structure $\rel I$ allows a homomorphism $\rel I \to \rel A$. Thus, $H$-colouring is the same as $\CSP(H)$. A key result in the complexity theory of CSP's (which is a culmination of decades of research and subsumes various previous results, including the Hell--Nešetřil dichotomy and an earlier one by Schaefer~\cite{Schaefer1978} for Boolean CSPs) is a general Dichotomy Theorem of Bulatov~\cite{Bul17} and Zhuk~\cite{Zhu20}, which asserts that for every finite relational structure $\rel A$, $\CSP(\rel A)$ is either \NP-complete, or solvable in polynomial time.
\emph{Promise constraint satisfaction problems} (\emph{PCSP}s) are a natural extension of CSPs, analogous to how the promise graph homomorphism problem extends the $H$-colouring problem.
This notion of PCSPs was introduced by Austrin, Guruswami, and Håstad~\cite{AGH17}, and the general theory was further developed by Brakensiek and Guruswami~\cite{BrakensiekG21}, and by Barto, Bulín, Krokhin, and Opršal~\cite{BBKO21}.

In addition to the aforementioned work \cite{KO19,WZ20,KOWZ23}, several other recent papers explore the connection between topology and the computational complexity of (P)CSPs. Schnider and Weber~\cite{SchniderW24} showed that Schaefer's dichotomy for the complexity of Boolean CSPs is reflected by the ``topological complexity'' of their solution spaces: the solution spaces of NP-complete problems are topologically arbitrarily complicated (in a precise technical sense), whereas the solution spaces of polynomial-time solvable problems are homotopy equivalent to a discrete set; this was further generalised to arbitrary CSPs by Meyer~\cite{Meyer24}. Meyer and Opršal~\cite{MeyerO24} gave a new, topological proof of the Hell–Nešetřil dichotomy, and Filakovský, Nakajima, Opršal, Tasinato, and Wagner~\cite{FilakovskyNOTW26} use topological methods related to the ones in the present paper to show that a certain hypergraph PCSP is \NP-hard. We believe that these results, and the ones presented here, are just the starting point of a promising line of research and that topological methods have the potential to yield further complexity-theoretic insights in the future.

\paragraph{Acknowledgement}
We are grateful to Andrei Krokhin, Marcin Wrochna, Standa Živný, and Libor Barto for valuable discussions during the early stages of this project. In particular, we would like to thank Marcin Wrochna and Libor Barto for sharing with us key observations that played an important role in the construction of a minion homomorphism to affine $\mathbb Z_2$-maps.
We also thank Sebastian Meyer for his feedback on the manuscript.

\section{Preliminaries}\label{sec:notation}

We denote the identity function on a set $A$ by $1_A$, we use the notation $[n] = \{1, \dots, n\}$, and the symbol $\sqcup$ for disjoint union.

\subsection{Polymorphisms and a Hardness Criterion}
  \label{sec:Polymorphisms}

We outline the fundamentals of the algebraic theory of PCSPs, in particular the core concept of \emph{polymorphisms} and the hardness criterion (Theorem~\ref{thm:hardness}) used in the proof of Theorem~\ref{thm:main}, focusing on special the case of graphs; see a survey by Krokhin and Opršal \cite{KO22} for a detailed treatment.

Given two graphs $G_1$ and $G_2$, their product $G_1 \times G_2$ is defined by $V(G_1\times G_2)=V(G_1) \times V(G_2)$ and $((u_1,u_2), (v_1,v_2))\in E(G_1\times G_2)$ if and only if $(u_1,v_1)\in E(G_1)$ and $(u_2,v_2)\in E(G_2)$; moreover, we denote by $G^n=G\times \dots \times G$ the product of $n$ copies of $G$.

\begin{definition}
  An $n$-ary \emph{polymorphism} from a graph $G$ to a graph $H$ is homomorphism $f\colon G^n\to H$, in other words, a map $f\colon V(G)^n \to V(H)$ such that $(f(u_1, \dots, u_n), f(v_1, \dots, v_n)) \in E(H)$ whenever $(u_1, v_1), \dots, (u_n, v_n) \in E(G)$.
  We denote\footnote{Somewhat unconventionally, we use lower-case notation for polymorphisms to highlight that we are not considering any topology on them, in contrast to the homomorphism complexes introduced below.}
  the set of all polymorphisms from $G$ to $H$ by $\pol(G, H)$, and the set of $n$-ary polymorphisms by $\pol^{(n)}(G, H)$.
\end{definition}

Polymorphisms are enough to describe the complexity of a promise CSP up to certain $\log$-space reductions \cite[Theorem 2.20]{KO22}. Loosely speaking, the more complex the polymorphisms are, the easier the problem is.
To formalise this, we define the notions of \emph{minor}, \emph{minion}, \emph{minion homomorphism}, and \emph{essential arity} which are necessary to formulate the hardness criterion.

Let $\pi\colon [n] \to [m]$, and let $A$ and $B$ be sets. The \emph{$\pi$-minor} of a function $f\colon A^n \to B$ is the function $f^\pi \colon A^m \to B$ given by $f^\pi (x_1, \dots, x_m) = f(x_{\pi(1)}, \dots, x_{\pi(n)})$ for all $x_1, \dots, x_m \in A$ (equivalently, if we view elements of $x\in A^n$ as functions $x\colon [n]\to A$, then $f^\pi (x) = f(x \circ \pi)$). A subset of the set of all functions $\{f\colon A^n\to B, n>0\}$ that is non-empty and closed under taking minors is called a \emph{function minion}. For example, it is easy to see that $\pol(G, H)$ has this property whenever $G$ and $H$ are graphs such that $G\to H$. Abstracting from this, we arrive at the following notion:

\begin{definition} \label{def:minion} \label{def:minion-homomorphism}
  An \emph{(abstract) minion} $\minion M$ is a collection of non-empty sets $\minion M^{(n)}$, where $n > 0$ is an integer, and mappings
  \[
    \pi^\minion M \colon \minion M^{(n)} \to \minion M^{(m)},
  \]
  for $\pi\colon [n] \to [m]$, which satisfy $\pi^\minion M \circ \sigma^\minion M = (\pi\circ \sigma)^\minion M$ whenever $\pi \circ \sigma$ is defined, and $(1_{[n]})^\minion M = 1_{\minion M^{(n)}}$.
  We will often write $f^\pi$ instead of $\pi^\minion M(f)$, and call this element the $\pi$-minor of $f$.

 A \emph{minion homomorphism} from a minion $\minion M$ to a minion $\minion N$ is a collection of maps $\xi_n\colon \minion M^{(n)} \to \minion N^{(n)}$ that preserve taking minors, i.e., such that for each $\pi\colon [n] \to [m]$, $\xi_m\circ \pi^\minion M = \pi^\minion N\circ \xi_n$.
We denote such a homomorphism simply by $\xi\colon \minion M \to \minion N$, and write $\xi(f)$ instead of $\xi_n(f)$ when the index is clear from the context.\end{definition}

Given a minion $\minion{M}$, an element $f \in \minion M^{(n)}$ is said to have \emph{essential arity at most $k$} if it is a minor of an element $g \in \minion M^{(k)}$. If there is a bound $N$, such that every element of $\minion M$ has essential arity at most $N$, $\minion M$ is said to have \emph{bounded essential arity}. An element $f\in \minion M^{(n)}$ is \emph{constant} if all its minors coincide, i.e., $f^\pi = f^\sigma$ for all $m > 0$ and $\pi, \sigma \colon [n] \to [m]$.
For example, in function minions, being constant coincides with the usual notion of being a constant function, and if a function $f\colon A^n\to B$ depends only on a subset of variables with indices $\{i_1, \dots, i_k\}$, then $f(x_1,\dots, x_n)= g(x_{i_1}, \dots, x_{i_k})$, so $f$ is of arity at most $k$.
Our proof uses the following hardness criterion.

\begin{theorem}[{\cite[Proposition 5.14]{BBKO21}}]\label{thm:hardness}
  Let $G$ and $H$ be two graphs such that $G \to H$.
  If there exists a minion homomorphism
  \[
    \xi\colon \pol(G, H) \to \minion B
  \]
  for some minion $\minion B$ of bounded essential arity which does not contain a constant, then $\PCSP(G, H)$ is \NP-complete.
\end{theorem}

\subsection{Topology and Homomorphism Complexes}
\label{sec:hom-complexes}

We review a number of topological notions that we will need in what follows, in particular the notion of \emph{homomorphism complexes}, a well-known construction in topological combinatorics that goes back to the work of Lovász \cite{Lovasz-Kneser}.
We refer the reader to \citet{Hat02} and  \citet{Mat03} for accessible general introductions to algebraic topology and topological combinatorics, respectively, and to \citet{Koz08} for an in-depth treatment of homomorphism complexes.

\paragraph{Simplicial sets}
In applications of topological methods in combinatorics and theoretical computer science, topological spaces are often specified combinatorially as simplicial complexes. For our purposes, it will be convenient to work instead with \emph{simplicial sets}, which generalize simplicial complexes in a way analogous to how directed multigraphs generalize simple graphs. Simplicial sets are a somewhat less common notion in topological combinatorics, but play an important role in homotopy theory, see \citet{Friedman2012} for a gentle combinatorial introduction.

Similarly to a simplicial complex, a simplicial set is a combinatorial description of how to build a space from vertices, edges, triangles, and higher-dimensional simplices. Informally speaking, we view the vertex set of each $n$-dimensional simplex as totally ordered (equivalently, labelled by $\{0,1,\dots,n\}$) and we are allowed to glue simplices together by linear maps between them that are given by (not necessarily strictly) monotone maps between their vertex sets. On the one hand, this permits more general glueings than in simplicial complexes (which allows constructing spaces using fewer simplices): for instance, we may glue both endpoints of an edge to the same vertex (creating a loop), or glue the endpoints of multiple edges to the same pair of vertices, or we may glue the the boundary of a triangle to a single vertex, forming a $2$-dimensional sphere $S^2$. On the other hand, the description is still purely combinatorial and, moreover, retains the information about the ordering of the vertices of each simplex before the glueing. This yields a natural notion of products of simplicial sets and will play an important role in the combinatorial arguments below.

\begin{definition}
A~\emph{simplicial set} $X$ is given by the following data: First, a collection of pairwise disjoint sets $X_0$, $X_1$, $X_2$, \dots; the elements of $X_n$ are called the \emph{$n$-simplices} of $X$. Second, for every pair of integers $m,n \geq 0$ and every (not necessarily strictly) monotone map $\alpha\colon \{0,\dots,m\} \to \{0,\dots,n\}$, there is a map $\alpha^X \colon X_n \to X_m$, such that $1_{\{0,\dots,n\}}^X = 1_{X_n}$ and such that $(\alpha \circ \beta)^X=\beta^X \circ \alpha^X$ whenever the composition is defined.
\end{definition}

Every simplicial set $X$ defines a topological space $\geom X$, the \emph{geometric realization} of $X$, which is obtained by glueing geometric simplices together according to the combinatorial data in $X$; we refer to \cite[Section~4]{Friedman2012} for a precise definition. We say that a simplicial set $X$ is a \emph{triangulation} of a topological space $T$ if $\geom X$ is homeomorphic to $T$. A $k$-simplex $\sigma \in X_k$ is called \emph{degenerate} if $\sigma = \alpha^X(\tau)$ for some $\tau \in X_m$ and $\alpha\colon \{0,\dots,k\} \to \{0,\dots,m\}$ with $m<k$. In the geometric realization, degenerate simplices are collapsed down to lower-dimensional simplices, and $\geom X$ is the disjoint union of the interiors of non-degenerate simplices;\footnote{In more technical terms, $\geom X$ is a CW complex with one $k$-cell for each non-degenerate $k$-simplex of $X$.}
however, degenerate simplices play an important role in specifying the glueings and the combinatorial data keeps track of them. All simplicial sets used in this paper have only finitely many non-degenerate simplices; this is equivalent to $\geom X$ being a compact space. The \emph{dimension} of a simplicial set $X$ is defined as the maximum dimension of a non-degenerate simplex of $X$.

A \emph{simplicial map} $f\colon X\to Y$ between simplicial sets is a collection of maps $f_n \colon X_n \to Y_n$, $n>0$, such that $f_m \circ \alpha^X = \alpha^Y \circ f_n$ for all monotone maps $\alpha\colon \{0,\dots,m\} \to \{0,\dots,n\}$. Every simplicial map $f\colon X\to Y$ defines a continuous map $\geom f\colon \geom X\to \geom Y$.
An \emph{isomorphism} of simplicial sets $X$ and $Y$ is a simplicial map $f\colon X\to Y$ with a simplicial inverse $g\colon Y \to X$ ($f_n$ is inverse to $g_n$ for all $n > 0$).

\paragraph{Products}
The product $X\times Y$ of two simplicial sets $X$ and $Y$ is the simplicial set
whose $n$-simplices of $X\times Y$ are ordered pairs $(\sigma,\tau)$, i.e., $(X\times Y)_n=X_n\times Y_n$), and $\alpha^{X\times Y}(\sigma,\tau)=(\alpha^X(\sigma),\alpha^X(\tau))$.
On the level of geometric realizations, this corresponds to the usual product of topological spaces, i.e., $\geom{X\times Y}\cong \geom X\times \geom Y$, under some mild conditions on $X$ and $Y$ that are satisfied for all simplicial sets we work with (e.g., if both $X$ and $Y$ are countable, see \cite[Theorem~5.2]{Friedman2012} for a general statement). The $n$th power of a simplicial set $X$ is $X^n=X \times \dots \times X$ (the product of $n$ copies of $X$).

\paragraph{Group actions} Various objects we work with in this paper (graphs, simplicial sets, topological spaces, etc.) have a natural symmetry given by an action of the cyclic group $\Ztwo$, which is described by a structure-preserving involution. For instance, a $\Ztwo$-action on a simplicial set $X$ is given by a simplicial map $\nu\colon X\to X$ that satisfies $\nu^2 \coloneq \nu \circ \nu = 1_X$ (thus, $\nu$ is necessarily a simplicial automorphism). We mainly work with actions that are \emph{free}, which for $\Ztwo$-actions simply means that $\nu$ has no fixed points. If $(X,\nu_X)$ and $(Y,\nu_Y)$ are simplicial sets with $\Ztwo$-actions, then a simplicial map $f\colon X\to Y$ is called \emph{equivariant} if it preserves the $\Ztwo$-symmetry, i.e., $f\circ \nu_X = \nu_Y \circ f$.

$\Ztwo$-actions on graphs (by isomorphisms) or on spaces (by homeomorphisms), and the notions of equivariant graph homomorphisms and equivariant continuous maps, etc., are defined analogously.

\paragraph{Relational simplicial sets} Most simplicial sets in this paper are of the following special form, which we call \emph{relational} (a non-standard term): The set $X_0$ of vertices ($0$-simplices) is a finite set, and $X_n \subset (X_0)^{n+1}$ is an $(n+1)$-ary relation, i.e., every $n$-simplex of $X$ is an ordered $(n+1)$-tuple $[u_0,\dots,u_n]$ of vertices (we use square brackets as a reminder that we view these $(n+1)$-tuples as simplices, and we identify each element $u\in X_0$ with the singleton tuple $[u]$). Moreover, for every monotone map $\alpha \colon \{0,\dots,m\} \to \{0,\dots,n\}$, the map $\alpha^X$ is defined by $\alpha^X([u_0,\dots,u_n])=[u_{\alpha(0)},\dots,u_{\alpha(m)}]$.
To get a simplicial set this way, the collection of relations $X_n$, $n>0$, needs to be closed under the operations $\alpha^X$, i.e., if $\sigma$ is a simplex of $X$, then any tuple obtained from $\sigma$ by omitting and/or repeating vertices without changing their order is a simplex as well.

\begin{example}[$\Ztwo$-symmetric triangulations of spheres]
\label{ex:spheres}
\begin{figure}
  \centering
  \includegraphics[scale=1.1]{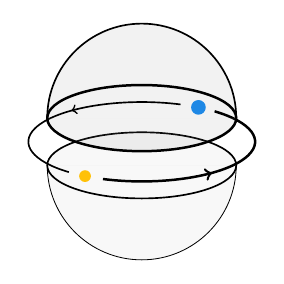}
  \caption{The simplicial set $\Sigma^2$} \label{fig:sigma-2}
  \Description{A decomposition of a sphere into two vertices, two edges between them following the equator, and two hemispheres, northern and southern.}
\end{figure}
We define a relational simplicial set $\Sigma^2$ that defines a triangulation of the $2$-dimensional sphere $S^2$, together with a natural $\Ztwo$-action that corresponds to the antipodal map $x\mapsto -x$ on $S^2$. The vertex set of $\Sigma^2$ is $\Sigma^2_0 = \{\yellowb, \blueb\}$ (which we think of as a pair of antipodal points in $S^2$), and $\Sigma^2_n$ is the set of all $(n+1)$-tuples of $\yellowb$ and $\blueb$'s with at most $2$ alternations. Thus, e.g., $[\yellowb, \blueb, \blueb, \yellowb]$ is a $3$-simplex of $\Sigma^2$, but $[\blueb, \yellowb, \blueb, \yellowb]$ is not. The $\Ztwo$-action on $\Sigma^2$ is given by the simplicial map that swaps the two vertices.

  This construction naturally generalises to yield a sequence of simplicial sets $\Sigma^0 \subset \Sigma^1 \subset \Sigma^2 \subset \dots$, such that $\Sigma^k$ (whose simplices are tuples with entries in $\{\yellowb, \blueb\}$ and at most $k$ alternations) is a triangulation of $S^k$. A simplex of $\Sigma^k$ is degenerate if and only if it contains two consecutive vertices of the same color. Thus, the only non-degenerate simplices of $\Sigma^0$ are the two vertices $\blueb, \yellowb$; $\Sigma^1$ additionally has two non-degenerate $1$-simplices $[\blueb, \yellowb]$ and $[\yellowb, \blueb]$ connecting these two vertices (geometrically, this corresponds to two distinct paths between a pair of antipodal points, each following half of an equatorial circle clockwise); $\Sigma^2$ adds two non-degenerate triangles $[\yellowb, \blueb, \yellowb]$ and $[\blueb, \yellowb, \blueb]$ which corresponds to glueing the northern and southern hemisphere, respectively (see Figure~\ref{fig:sigma-2}); $\Sigma^3$ adds two non-degenerate 3-simplices; etc.
\end{example}

\begin{obs}
\label{obs:2-coloring}
  If $X$ is a (relational)\footnote{The claim is true for arbitrary simplicial set $X$; it is only required that $\Sigma^2$ is relational.}
    simplicial set then a simplicial map $X\to \Sigma^2$ is completely described by a $2$-colouring of the vertex set $X_0$ with colours yellow or blue. Conversely, a $2$-colouring $f$ of $X_0$ defines a simplicial map if and only if there is no $3$-simplex $[u_0, u_1, u_2, u_3]$ of $X$ such that $[f(u_0), f(u_1), f(u_2), f(u_3)]$ has three alternations (is equal to either $[\blueb, \yellowb, \blueb, \yellowb]$ or $[\yellowb, \blueb, \yellowb, \blueb]$). Moreover, if $\Ztwo$-acts on $X$ by a simplicial involution $\nu$, then such a $2$-colouring defines an equivariant map if and only if $u$ and $\nu(u)$ have different colours for every vertex $u$ of $X$.
\end{obs}

\paragraph{Order complexes of posets} Another important example of relational simplicial sets are order complexes: Given a finite partially ordered set (poset) $P$, the \emph{order complex} $\Delta(P)$ is the simplicial set whose $n$-simplices are weakly monotone chains, i.e., $(n+1)$-tuples $[u_0,\dots,u_n] \in P^{n+1}$ with $u_0\leq \dots \leq u_n$; moreover, for every monotone map $\alpha\colon \{0,\dots,m\}\to \{0,\dots,n\}$, $\alpha^{\Delta(P)}[u_0,\dots,u_n]=[u_{\alpha(0)},\dots,u_{\alpha(m)}]$ as above. Note that monotonicity of $\alpha$ is crucial here to ensure that chains are mapped to chains.  An $n$-simplex $[u_0,\dots,u_n]$ of $\Delta(P)$ is non-degenerate if and only if $u_0 < \dots < u_n$.

Any monotone map $f\colon P\to Q$ between posets naturally extends componentwise to chains and hence to a simplicial map $f\colon \Delta(P) \to \Delta(Q)$ between order complexes.

\begin{example} \label{ex:GammaL}
Let $L$ be a positive integer divisible by $4$. Define a partial order $\preccurlyeq$ on $\Z_L=\{0,1,\dots,L-1\}$ by $a \prec b$ if and only if $a$ is even, $b$ is odd, and $a-b = \pm 1 \mod L$. We define the simplicial set $\Gamma_L$ as the order complex of this poset,
\[
\Gamma_L := \Delta( \Z_L,\preccurlyeq)
\]
The simplicial set $\Gamma_L$ is a triangulation of $S^1$, see Figure~\ref{fig:box-k3} (as a digraph, it is a cycle of length $L$ with edges oriented alternatingly). Moreover, the map $\Z_L \to \Z_L$, $x\mapsto x+L/2$ defines a simplicial involution $\Gamma_L\to \Gamma_L$ that corresponds to the antipodal involution on $S^1$.
\end{example}

If $P$ and $Q$ are posets and if we consider the product $P\times Q$ with the componentwise partial order $(p,q) \leq (p',q')$ if and only if $p\leq p'$ and $q\leq q'$, then $\Delta(P\times Q)$ and $\Delta(P)\times \Delta(Q)$ are isomorphic simplicial sets. In particular, $\Gamma_L^n = \Gamma_L \times \dots \times \Gamma_L$ is a triangulation of the $n$-dimensional torus $T^n=S^1 \times \dots \times S^1$.
Note that the vertices $\Gamma_L^n$ are $n$-tuples $\boldsymbol{u} \in \Z_L^n$, and $k$-simplices are $(k+1)$-tuples of vertices $[\boldsymbol u_0, \dots, \boldsymbol u_k]$ such that $\boldsymbol u_{i+1}$ is obtained from $\boldsymbol u_i$ by choosing a subset of coordinates of $\boldsymbol u_i$ all that are even and changing each of them by $\pm 1$ modulo $L$.

\paragraph{Homomorphism complexes}
\label{par:homomorphism-complexes}
Given two graphs $F$ and $G$, the homomorphism complex $\Hom(F,G)$ is a simplicial set capturing the structure of all homomorphisms $F\to G$. Following \citet[Section~5.9]{Mat03}, we define homomorphism complexes as order complexes of the poset of \emph{multihomomorphisms} from $F$ to $G$.\footnote{We remark that in \cite{Mat03} order complexes are defined as simplicial complexes, but the two definitions are equivalent. There are several other alternative definitions of homomorphism complexes that lead to topologically equivalent spaces. }
By definition, a \emph{multihomomorphism} is a function $f\colon V(F) \to 2^{V(G)} \setminus \{\emptyset\}$ such that, for all edges $(u, v) \in E(F)$, we have that
\[
  f(u) \times f(v) \subseteq E(G).
\]
We denote the set of all multihomomorphisms by $\mhom(F, G)$. Multihomomorphisms are partially ordered by component-wise inclusion: $f \leq g$ if and only if $f(u) \subseteq g(u)$ for all $u\in V(F)$.

\begin{definition} Let $F$ and $G$ be graphs. The \emph{homomorphism complex} $\Hom(F,G)$ is the order complex $\Delta(\mhom(F, G), \leq)$ of the poset of multihomomorphisms.
\end{definition}

Multimorphisms can be composed in a natural way: if $f\in \mhom(F, G)$ and $g\in \mhom(G,H)$, then $(g\circ f)(a) = \bigcup_{b\in f(a)} g(b)$ is a multihomomorphism from $F$ to $H$. In particular, every homomorphism $f\colon G\to H$ induces a simplicial map $f_* \colon \Hom(F, G) \to \Hom(F, H)$ defined on vertices by mapping a multihomomorphism $m\in \mhom(F,G)$ to the composition $f\circ m$.

In what follows, we will focus on the special case of $\Hom(K_2, G)$, a common tool in the study of graph colourings.
Note that a multimorphism $m$ from $K_2$ to a graph $G$ corresponds to an ordered pair of subsets $m(1), m(2) \subseteq V(G)$ such that any pair of vertices $v_1 \in m(1)$ and $v_2\in m(2)$ are connected by an edge. If $G$ has no loops, then $m(1)$ and $m(2)$ are disjoint and induce a complete bipartite subgraph of $G$.
The natural $\Ztwo$-action on $K_2$ that swaps the two vertices induces an induces an action on multihomomorphisms $m\colon K_2 \to G$, namely swapping the two sets $m(1)$ and $m(2)$, which in turn induces a $\Ztwo$-action on the simplicial set $\Hom(K_2, G)$; this action is free provided $G$ has no loops. Moreover, it is easy to check that for every graph homomorphism  $f\colon G\to H$, the induced simplicial map $f_* \colon \Hom(K_2, G) \to \Hom(K_2, H)$ is equivariant.

The following two examples will play an important role in this paper.

\begin{figure}
  \[
    \begin{tikzpicture}[
    baseline=(current bounding box.center),
    every node/.style={circle, inner sep=1pt},
    decoration={markings, mark=at position 0.5 with {\arrow{>}}}]
  
  \node (20) at (0:2cm)      {0};
  \node (12/0) at (30:2cm)   {1};
  \node (10) at (60:2cm)     {2};
  \node (1/02) at (-270:2cm) {3};
  \node (12) at (-240:2cm)   {4};
  \node (01/2) at (-210:2cm) {5};
  \node (02) at (-180:2cm)   {6};
  \node (0/12) at (-150:2cm) {7};
  \node (01) at (-120:2cm)   {8};
  \node (02/1) at (-90:2cm)  {9};
  \node (21) at (-60:2cm)    {10};
  \node (2/01) at (-30:2cm)  {11};

  \draw[postaction={decorate}] (21) -- (02/1);
  \draw[postaction={decorate}] (10) -- (12/0);
  \draw[postaction={decorate}] (02) -- (01/2);
  \draw[postaction={decorate}] (01) -- (02/1);
  \draw[postaction={decorate}] (02) -- (0/12);
  \draw[postaction={decorate}] (01) -- (0/12);
  \draw[postaction={decorate}] (12) -- (1/02);
  \draw[postaction={decorate}] (10) -- (1/02);
  \draw[postaction={decorate}] (21) -- (2/01);
  \draw[postaction={decorate}] (20) -- (12/0);
  \draw[postaction={decorate}] (20) -- (2/01);
  \draw[postaction={decorate}] (12) -- (01/2);
\end{tikzpicture}
   \]
  \caption{The simplicial set $\Gamma_{12}\cong \Hom(K_2,K_3)$; see also Examples~\ref{ex:GammaL} and \ref{ex:odd-cycles}.}
  \Description{A 12-gon with vertices $0, 1, \dots, 11$ with alternating orientation of edges.}
\end{figure}
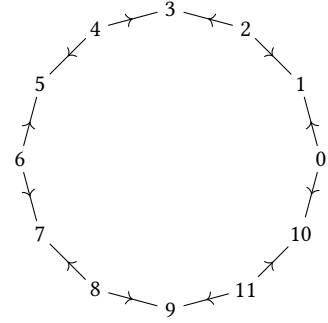

\begin{figure*}
  \centering
    \includegraphics{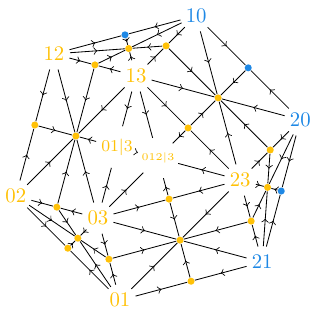}\qquad
    \includegraphics{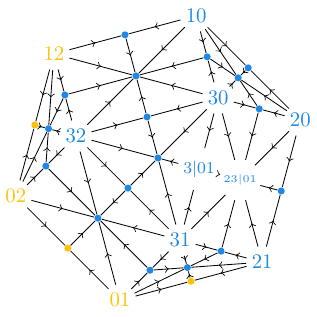}
  \caption{The simplicial set $\Hom(K_2, K_4)$; see also Example~\ref{ex:box-k4}.}
  \label{fig:box-k3} \label{fig:box-k4}
  \Description{Cubo-octahedron whose each face is subdivided using the barycentric subdivision.}
\end{figure*}

\begin{example} \label{ex:odd-cycles}
  For every odd integer $\ell\geq 3$, $\Hom(K_2, C_\ell)$ is isomorphic to the simplicial set $\Gamma_{4\ell}$ defined above; moreover, this isomorphism is equivariant, i.e., it preserves the $\Ztwo$-action.
\end{example}

\begin{example} \label{ex:box-k4}
  The simplicial set $\Hom(K_2, K_4)$ is a triangulation of a sphere $S^2$; it is depicted in Figure~\ref{fig:box-k4} which shows two hemispheres of this sphere that are glued together along their boundary. The homomorphisms/edges are explicitly labelled (the edge $(u, v)$ is labelled by $uv$) to highlight the global structure, and a few multihomomorphisms are labelled (where $3|01$ denotes the multihomomorphism $0 \mapsto 3$ and $1\mapsto \{0, 1\}$, etc.) to explain how the triangles are constructed.
\end{example}

\begin{lemma}
\label{lem:simplicial-map-HomK4-Sigma2}
There exists an equivariant simplicial map
  \[
    \maptosigma \colon \Hom(K_2, K_4) \to \Sigma^2.
  \]
\end{lemma}
\begin{proof}
By Observation~\ref{obs:2-coloring}, such a map is given by a suitable $2$-colouring of the vertices of $\Hom(K_2, K_4)$. One suitable $2$-colouring is depicted in Figure~\ref{fig:box-k4}.
\end{proof}

\paragraph{Homotopy}\label{sec:homotopy}

Two continuous maps $f,g\colon X \to Y$ between topological spaces are \emph{homotopic}, denoted $f \sim g$, if there is a continuous map $h\colon X\times [0, 1] \to Y$ such that $h(x, 0) = f(x)$ and $h(x, 1) = g(x)$; the map $h$ is called a \emph{homotopy} from $f$ to $g$. Note that a homotopy can also be thought of as a family of maps $h({\cdot},t)\colon X\to Y$ that varies continuously with $t \in [0,1]$. In what follows, $X$ and $Y$ will often be given as simplicial sets, but we emphasize that we will generally not assume that the maps (or homotopies) between them are simplicial maps.
Two spaces are \emph{homotopy equivalent} if there are continuous maps $f \colon X \to Y$ and $g\colon Y \to X$ such that $f \circ g \sim 1_Y$ and $g \circ f \sim 1_X$.

These notions naturally generalize to the setting of spaces with group actions:
Two equivariant maps $f, g\colon X \to Y$ between spaces with $\Ztwo$-actions are \emph{equivariantly homotopic}, denoted by $f\sim_{\Ztwo} g$, if there exists an \emph{equivariant homotopy} between them, i.e., a homotopy $h\colon X \times [0,1] \to Y$ such that all maps $h(\cdot,t)\colon X\to Y$ are equivariant. We denote by $[X, Y]_\Ztwo$ the set of all equivariant maps $X\to Y$ up to equivariant homotopy, i.e.,
\[
  [X, Y]_{\Ztwo} = \{ [f] \mid f\colon X \to Y \text{ is equivariant} \},
\]
where $[f]$ denotes the set of all equivariant maps $g$ s.t.\ $f \sim_\Ztwo g$.

\section{Overview of the Proof}
\label{sec:overview}

We present a detailed overview of the proof of Theorem~\ref{thm:main}.
Every non-bipartite, loopless graph $G$ contains a cycle $C_\ell$ of odd length $\ell \geq 3$. In particular, there exists a homomorphism $C_\ell \to G$, hence every $C_\ell$-colourable graph is $G$-colourable. This yields a trivial reduction from $\PCSP(C_\ell, K_4)$ to $\PCSP(G, K_4)$.
Thus, Theorem~\ref{thm:main} follows from the following:

\begin{theorem}\label{thm:main-cycles}
For all odd integers $\ell \geq 3$, the decision problem $\PCSP(C_\ell, K_4)$ is NP-hard.
\end{theorem}

We will prove this using Theorem~\ref{thm:hardness}; to this end, we need to construct a minion homomorphism from $\pol(C_\ell,K_4)$ to a minion $\minion B$ that contains no constant and is of bounded essential arity.

Informally speaking, as mentioned in Section~\ref{sec:Polymorphisms}, the general philosophy of the algebraic approach to PCSPs is that in order to understand the complexity of a problem, we need to get a good-enough understanding of the structure of its polymorphisms, in our case, the structure of all graph homomorphisms $C_\ell^n \to K_4$, $n>0$, i.e., $4$-colourings of powers of an odd cycle.
Prima facie, such colourings do not seem to have any apparent structure, so we use topology to simplify the problem and reveal more information.
In the first step, using homomorphism complexes, we pass from the problem of understanding graph homomorphisms to the problem of understanding equivariant homotopy classes of equivariant continuous maps $T^n \to S^2$. This provides an approximation of the structure of polymorphisms, nevertheless classifying such continuous maps is still difficult (this is connected to the fact $S^2$ has many non-trivial higher \emph{homotopy groups} $\pi_k(S^2)$, $k \geq 3$).
Thus, in a second step, we replace $S^2$ by a ``topologically simpler'' space $\eml$. We can then quite explicitly describe, in a third step, the set of $[X, \eml]_\Ztwo$ in terms of a suitable (equivariant) cohomology group (using  \emph{equivariant obstruction theory}); this yields a minion homomorphism $\phi$ from $\pol(C_\ell, K_4)$ to a minion $\minion Z_2$ (defined precisely below). The fact that all maps and homotopies are equivariant ensures that the minion $\minion Z_2$ does not contain any constants; however, it is still not of bounded essential arity. In a fourth step, we then argue that the image of $\phi$ actually is of bounded essential arity, for which we use some of the previously neglected combinatorial structure. We now describe these steps in more detail:

\paragraph{Step~1} If $X$ and $Y$ are simplicial sets with $\Ztwo$-actions, then the set of of all equivariant simplicial maps $X^n \to Y$, $n>0$, is closed under taking minors, i.e., it forms a minion, which we denote by $\spol(X,Y)$ (this follows easily from the definition of products of simplicial sets).

In the first step of the construction, we use homomorphism complexes to associate with every graph homomorphism  $f \colon C_\ell^n  \to K_4$ an equivariant simplicial map $\mu(f)\colon \Gamma_{4\ell}^n \to \Sigma^2$, where $\Gamma_{4\ell}$ and $\Sigma^2$ are the simplicial sets described in Examples~\ref{ex:GammaL} and \ref{ex:spheres}, respectively. The simplicial map $\mu(f)$ is defined as a composition $\maptosigma \circ f_* \circ \iota_n$:
\[
  \Gamma_{4\ell}^n \cong \Hom(K_2,C_\ell)^n
  \overset{\iota_n}{\to} \Hom(K_2,C_\ell^n)
  \overset{f_*}{\to} \Hom(K_2,K_4)
  \overset{s}{\to} \Sigma^2,
\]
where $f_* \colon \Hom(K_2,C_\ell^n) \to \Hom(K_2,K_4)$ is the simplicial map induced by $f$, $\maptosigma \colon \Hom(K_2,K_4) \to \Sigma^2$ is the simplicial map from Lemma~\ref{lem:simplicial-map-HomK4-Sigma2}, the isomorphism $\Gamma_{4\ell}^n \cong \Hom(K_2,C_\ell)^n$ is given by the isomorphism from Example~\ref{ex:odd-cycles}, and the simplicial map $\iota_n$ is given by the special case $G=C_\ell$ of the following fact:

\begin{lemma}
  For every graph $G$ and $n\geq 1$, there is an equivariant simplicial map
  \[
    \iota_n \colon \Hom(K_2,G)^n \to \Hom(K_2,G^n).
  \]
\end{lemma}

\begin{proof}
  Given an $n$-tuple $m=(m_1,\dots,m_n)$ of multihomomorphisms $m_i\colon K_2\to G$, we can view $m$ as a multihomomorphism $\iota_n(m)\colon K_2 \to G^n$ by setting $\iota_n(m)(u)=m_1(u)\times \dots  \times m_n(u)$ for each vertex $u$ of $K_2$. This yields a map $\iota_n \colon \mhom(K_2,G)^n \to \mhom(K_2,G^n)$ that is monotone and equivariant and hence extends to the desired simplicial map.\footnote{It is easy to see that $\iota_n$ is injective, though generally not surjective, and it is known \cite[Proposition 18.17]{Koz08} that $\iota_n$ defines an equivariant homotopy equivalence between the spaces $\geom{\Hom(K_2,G)}^n$ and $\geom{\Hom(K_2,G^n)}$, but we will not need this fact in what follows.}
\end{proof}

The assignment $f\mapsto \mu(f)$ defines a map $\mu\colon \pol(C_\ell,K_4) \to \spol(\Gamma_{4\ell},\Sigma^2)$ that preserves arity. The map $\mu$ does \emph{not} strictly speaking preserve minors, i.e., for a general function $\pi\colon [n] \to [m]$, the simplicial maps $\mu(f)^\pi$ and $\mu(f^\pi)$ need not be equal, but it is not hard to see that the induced continuous maps are equivariantly homotopic. Thus, if we denote by $[\mu(f)] \in [T^n,S^2]_{\Ztwo}$ the equivariant homotopy class of the map $\geom{\mu(f)} \colon T^n\cong \geom{\Gamma_{4\ell}^n} \to \geom{\Sigma^2}\cong S^2$, then $[\mu(f)^\pi] = [\mu(f^\pi)]$ (see Lemma~\ref{lem:mu}).

\paragraph{Step~2} Determining the set of equivariant homotopy classes of maps $[T^n, S^2]_\Ztwo$ is a difficult problem (and closely related homotopy-theoretic questions regarding maps $X\to S^2$ for spaces of dimension $\dim X\geq 4$ are algorithmically undecidable \cite{Cadek:2013}). We circumvent this difficulty by enlarging $\geom{\Sigma^2}\cong S^2$ to a larger $\Ztwo$-space $\eml$ that is ``homotopically simpler'' (in technical terms, $\eml$ is an \emph{Eilenberg--MacLane space}), which makes $[T^n,\eml]_{\Ztwo}$ much easier to compute.

Given a simplicial map $g\colon \Gamma_{4\ell}^n \to \Sigma^2$, we define $\eta(g) \in [T^n, \eml]_{\Ztwo}$ as the equivariant homotopy class of the composition of the geometric realization $\geom{g}\colon T^n \to S^2$ with the inclusion map $j\colon S^2\hookrightarrow \eml$. It is easy to show that $\eta$ preserves minors, and hence defines a minion homomorphism from $\spol(\Gamma_{4\ell},\Sigma^2)$ to the minion $\hpol{S^1, \eml}$ of equivariant homotopy classes of equivariant maps, i.e., the minion with $\hpol[n]{S^1, \eml} = [T^n, \eml]_\Ztwo$, where $T^n = (S^1)^n$, and minors defined in the natural way.

By considering the composition $\phi \coloneq \eta\circ \mu$ with the map constructed in Step~1, we get the following:

\begin{lemma}[Section~\ref{app:minion-homomorphism}]
  \label{lem:minion-homomorphism}
  There are minion homomorphisms $\phi\colon \pol(C_\ell, K_4) \to \hpol{S^1, \eml}$ and $\eta\colon \spol(\Gamma_{4\ell}, \Sigma^2)\to \hpol{S^1, \eml}$ such that $\im \phi \subseteq \im \eta$, i.e., for each polymorphism $f\colon C_\ell^n \to K_4$, there is a simplicial map $g\colon \Gamma_{4\ell}^n \to \Sigma^2$ with $\phi(f) = \eta(g)$.
\end{lemma}

\paragraph{Step~3} Next, we give an explicit description of the sets $[T^n,\eml]_{\Ztwo}$.
This description is by the means of functions $f_{\alpha}\colon \Ztwo^n \to \Ztwo$ of the form $f_{\alpha}(x_1,\dots,x_n) = \sum_{i=1}^n \alpha_i x_i$, where $\alpha = (\alpha_1,\dots, \alpha_n) \in \Z_2^n$ and $\sum \alpha_i \equiv 1 \pmod 2$. For a fixed $n$, the set of such functions forms an affine space, which we denote by $\minion Z_2^{(n)}$, and together, these sets form a function minion $\minion Z_2$. Below, we will often identify an affine function $f_\alpha$ with the corresponding $n$-tuple $\alpha \in \Ztwo^n$ of coefficients, i.e., we will often view $\minion Z_2$ as an abstract minion, with $\minion Z_2^{(n)} =\{\alpha \in \Ztwo^n\colon \sum_i \alpha_i \equiv 1\bmod 2\}$.

\begin{proposition}[Section~\ref{app:monomials-degrees}]
\label{thm:affine_minion}
  For each $n > 0$, there is a bijection
  \[
    \gamma_n \colon [T^n,\eml]_{\Ztwo} \to \minion Z_2^{(n)}.
  \]
  Moreover, they preserve minors, and hence form a minion isomorphism $\gamma \colon \hpol{S^1, \eml} \to \minion{Z}_2$.
\end{proposition}

The proof of this proposition has two parts. One the one hand, using  \emph{equivariant obstruction theory}, we can prove the following:
\begin{lemma}[Section~\ref{app:cohomology}]
\label{lem:number-homotopy-classes}
The set $[T^n,\eml]_{\Ztwo}$ has cardinality $2^{n-1}$.
\end{lemma}
One the other hand, every $\alpha \in \minion{Z}_2^{(n)}$ corresponds to a square-free monomial $\prod_{i\in I} z_i$ of odd degree in the variables $z_1,\dots,z_n$, where $I = \{i\in [n]\colon \alpha_i = 1\}$. If we view $S^1=\{z\in \mathbb{C} \colon |z|=1\}$ as the unit circle in the complex plane then each such monomial gives rise to an equivariant map $T^n=(S^1)^n \to S^1$ given by $(z_1,\dots,z_n)\mapsto \prod_{i\in I} z_i$. By composing first with a fixed equivariant inclusion $S^1\hookrightarrow S^2$ (e.g., the one given by the inclusion $\Sigma^1\subset \Sigma^2$) and then with the inclusion $j\colon S^2\hookrightarrow \eml$, we can also view each such monomial $\prod_{i\in I} z_i$ as an equivariant map $m_\alpha\colon T^n\to \eml$. Using a geometrically defined set of $\Ztwo$-valued invariants $\deg_i$, $1\leq i\leq n$, we will show in Section~\ref{app:monomials-degrees} that these maps are pairwise non-homotopic; in fact, we will see that the map $\gamma_n \colon [T^n,\eml]_{\Ztwo} \to \minion Z_2^{(n)}$ defined by $\gamma_n([f])=(\deg_1(f),\dots,\deg_n(f))$ satisfies $\gamma_n(m_\alpha)=\alpha$. Thus, $\gamma_n$ is surjective, and hence bijective, by Lemma~\ref{lem:number-homotopy-classes}; therefore, every equivariant map $T^n\to Y$ is equivariantly homotopic to a unique monomial map $m_\alpha$ with $\alpha \in \minion Z_2^{(n)}$. Moreover, we will show that the maps $\gamma_n$ reserve minors, hence they form a 
minion isomorphism.

\paragraph{Step~4} Finally, we show (Section~\ref{sec:bounded-arity}) that for every equivariant simplicial map $f \colon \Gamma_{4\ell}^n \to \Sigma^2$, the equivariant homotopy class $\eta(f)\in [T^n,\eml]_{\Ztwo}$ corresponds to an odd monomial map $\prod_{i\in I} z_i$ with $\abs I = O(\ell^2)$.
This is proved by a combinatorial averaging argument, using the structure of the triangulation $\Gamma_{4\ell}^n$, the fact that simplicial maps to $\Sigma^2$ correspond to vertex $2$-colourings without alternating $3$-simplices, and the geometric definition of the invariants $\deg_i$. Thus, the image of $\pol(\Gamma_{4\ell}, \Sigma^2)$ under $\eta$, and hence the image of $\pol(C_\ell,K_4)$ under $\phi$, has bounded essential arity.
This concludes the proof of Theorem~\ref{thm:main-cycles}.

\paragraph{Comparison with previous work} The topological approach in \cite{KOWZ23} for proving hardness of $\PCSP(C_\ell, K_3)$, on which our work builds, required understanding the structure of the set equivariant maps from $T^n$ to $S^1$ up to equivariant homotopy. Such maps can be classified by much more elementary arguments using fundamental groups and winding numbers, which show that $[T^n,S^1]_{\Ztwo}$ is isomorphic to the affine space of maps $\Z^n \to \Z$ of the form $(x_1,\dots,x_n)\mapsto \sum_i \alpha_i x_i$, where $a_i \in \Z$ and $\sum_i a_i \equiv 1 \mod 2$ (this implicitly exploits the fact that $S^1$ is already an Eilenberg--MacLane space, i.e., has trivial higher homotopy groups).

Moreover, bounding the essential arity of such maps that arise from graph homomorphisms is also relatively simple: by considering suitable simplicial embeddings of $\geom{\Gamma_{4\ell}} \cong S^1$ into $T^n$, the sum $\sum_i |a_i|$ of absolute values of coefficients in such a map can be read of as the winding number of a simplicial map $\Gamma_{4\ell} \to \Sigma^1$, hence is $O(\ell)$.
By contrast, the more careful counting argument required in our case, although elementary in hindsight, was elusive for several years.

Moreover, although the method based on equivariant obstruction theory was developed to address approximate graph colouring, it first found an application in a hardness proof for promise \emph{linearly-ordered colouring} of hypergraphs \cite{FilakovskyNOTW26}, which uses a simpler hardness criterion, and allows for an easier combinatorial argument bounding the arity.

\begin{figure}
  \includegraphics[width=.8\textwidth]{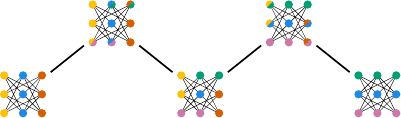}
\caption{A path in $\mhom(K_3^2, K_5)$.}
  \label{fig:x=y}
\end{figure}

\paragraph{Limitations of our approach}

We comment on the obstruction to generalising this result to colouring with 5 or more colours.
The immediate generalisation of our proof fails at Step~2, where we factor polymorphisms by homotopy, i.e., replace $\spol$ with $\operatorname{hpol}$.
The problem is that the homomorphism complex $\Hom(K_3^2, K_5)$, and hence also $\Hom(C_{2\ell+1}^2, K_5)$ for all $\ell > 1$, is connected. In particular, observe that two unary maps $f, g\colon K_3^2 \to K_5$ depending on two different coordinates induce homotopic maps $f_*, g_* \colon T^2 \to S^3$; one such homotopy is depicted in Figure~\ref{fig:x=y} which shows a path between the two homomorphisms in the multihomomorphism poset.
Consequently, this construction of a minion homomorphism does not avoid a constant operation as is required by all (known) hardness criteria for minions.

There are two possible way how to circumvent this problem.
Firstly, it should be noted that $\Hom(K_3^2, K_5)$ is connected, but not contractible (nor is the space of equivariant maps $T^2 \to S^3$). Hence a possible approach to generalising our result could take advantage of the higher homotopy structure.
Secondly, the same issue also arises for the problem of promise linearly-ordered hypergraph colouring $\PCSP(\mathbf{LO}_2, \mathbf{LO}_3)$ studied by Krokhin and Vagnozzi \cite{KrokhinV26}; more precisely, the space $\Hom(\mathbf{LO}_2^2, \mathbf{LO}_3)$ is also connected but not contractible.
Our proof presented here can be viewed as assigning a set $I(f)$ of `important coordinates' to each polymorphism $f$; these sets are of bounded size, and play well with taking minors.  Krokhin and Vagnozzi were able to prove the hardness of $\PCSP(\mathbf{LO}_2, \mathbf{LO}_3)$ in a similar way: If a polymorphism $f$ belongs to a specific subset $S\subseteq \pol(\mathbf{LO}_2, \mathbf{LO}_3)$, one can figure $I(f)$ directly using a specific combinatorial structure. For polymorphisms $f$ outside of this set $S$, they then use a topological argument to compute $I(f)$: Removing all polymorphisms within $S$ from the space which makes the space disconnected, and each conneted component contains exactly one projection, hence $I(f) = \{i\}$ where $i$ is the corresponding coordinate.

\paragraph{Organisation of the proof}

We present the full proof in the following four sections starting with the two key technical steps.
An overview of the definitions of topological invariants used in Step~3 is given in Section~\ref{app:monomials-degrees}, and the combinatorial argument of Step~4 in Section~\ref{sec:bounded-arity}. The subsequent two sections then outline the necessary general theory. Section~\ref{app:minion-homomorphism} deals with a construction of minion homomorphism required for Step~2, and Section~\ref{app:cohomology} describes the construction and properties of the space $Y$ used in Step~3.
 
\section{\texorpdfstring{Monomial Maps, Degrees, and $[T^n, \eml]_{\Z_2}$}{Monomial Maps, Degrees, and Homotopy Classes of Maps}}
\label{app:topology}
\label{app:monomials-degrees}

The goal of this section is to prove Proposition~\ref{thm:affine_minion}. 

To this end, we will define, for every equivariant continuous map $f\colon T^n\to Y$, a sequence of numbers $\deg_i(f) \in \Ztwo$, $1\leq i \leq n$, that are invariant under equivariant homotopy. As we will see below, these numbers satisfy $\sum_{i=1}^n \deg_i(f)\equiv 1 \bmod 2$. Thus, by assigning to every equivariant homotopy class $[f] \in [T^n, \eml]_{\Z_2}$ the sequence $\gamma_n([f])=\left(\deg_1(f),\dots,\deg_n(f)\right)\in \Ztwo^n$, we get a well-defined map 
$\gamma_n \colon  [T^n, \eml]_{\Z_2} \to \affine_2^{(n)}$. 

To define the invariants $\deg_i$ and throughout this section, we assume some familiarity with fundamental notions of algebraic topology, including homotopy, CW complexes, simplicial and cellular approximation theorems, and simplicial and cellular homology and cohomology; we refer to \citet{Hat02} for general background, and to May et al.~\cite[Chapters I and II]{MayCP96},\citet{Die87}, and \citet{Bre67} for more details on the equivariant setting.

We will use the fact that the space $Y$ is is a CW complex constructed from $\Sigma^2$ by attaching higher-dimensional cells (see the proof of Lemma~\ref{lem:Postnikov}); in particular, the $1$-dimensional and $2$-dimensional skeleta of $Y$ are $\Sigma^1$ and $\Sigma^2$, respectively. 

For a CW complex $X$, let $C_\bullet(X)$ and $C^\bullet(X)$ denote the cellular chain and cochain complexes of $X$ with $\Ztwo$-coefficients, respectively (since we work with $\Ztwo$-coefficients, $i$-dimensional cochains correspond to subsets of $i$-dimensional cells of $X$, and $i$-dimensional chains correspond to finite subsets); in the special case that $X$ is a simplicial complex or simplicial set, $C_\bullet(X)$ and $C^\bullet(X)$ are isomorphic to the simplicial chain and cochain complex of $X$, respectively.

To motivate the following definition, consider 
the torus $T^2\cong \lvert\Gamma^2_L\rvert$ and and equivariant map $f\colon T^2  \to Y$. Consider a loop in $T^2$ that wraps around the first coordinate direction, say the circle $x_1 = \{(z_1, 1)\colon z_1\in S^1\} \cong S^1 \subset T^2$. Note that the circle $x_1$ is triangulated by a subcomplex of $\Gamma^2_L$, but it is not fixed under the $\Ztwo$-action on $T^2$. If $f=m_{(1,0)}$ is the monomial map given by $(z_1,z_2)\mapsto \mapsto j(z_1)$, then $f$ maps the circle $x_1$ to the $1$-skeleton $\Sigma^1$ of $Y$, which is a circle as well, and as a map between circles, $f$ has degree $1$; thus, there is an odd number of edges $[v,w]$ in the triangulation of $x_1$ that satsify $f(v) = \blueb$ and $g(w) = \yellowb$. If $f$ is merely equivariantly homotopic to $m_{(1,0)}$, however, then this need no longer be the case: Intuitively, we can visualize the homotopy as ``moving'' the image of the edges of the torus through the discs in $\Sigma^2$, therefore potentially changing the parity of the degree we are interested in. The homotopy has to be equivariant, however, and thus has to modify each antipodal edge in the opposite way. As a consequence, a $2$-dimensional band connecting the circle $x_1$ and its antipodal circle $\nu\cdot x_1$ has to be ``dragged'' around over the discs of $\Sigma^2$ and thus, while the degree might change along equivariant homotopies, this difference will be registered in the behaviour of the connecting band. We will now formalize this geometric intuition.

\begin{definition} 
Let $L, L'$ be two positive integers divisible by $4$, and consider the $2$-dimensional torus $T^2=|\Gamma_{L}\times \Gamma_{L'}|$. 

Let $e^0\in C^1(\eml)$ be the dual of $e_0 = [\blueb, \yellowb]\in C_1(\Sigma^2) = C_1(\eml)$ and $d^0\in C^2(\eml)$ the dual of $d_0 = [\blueb, \yellowb, \blueb]\in C_2(\Sigma^2) = C_2(\eml)$ (i.e., $e^0([\yellowb, \blueb]) = 0$ and $e^0([\blueb, \yellowb]) = 1$, similarly for $d^0$). 
Moroever, let $x_1\in Z_1(\Gamma_L \times \Gamma_{L'})$ be the ``first coordinate cycle'' in $\abs{\Gamma_{L}\times\Gamma_{L'}}\cong T^2$ (i.e., $x_1 = \sum_{k=0}^{L-1} [(k, 0)(k+1, 0)]$ in $Z_1(\Gamma_{L}\times\Gamma_{L'})$), and let $b_1 \in C_2(\Gamma_L \times \Gamma_{L'})$ be the ``band'' connecting $x_1$ with $\gromega\cdot x_1$ (i.e., $\partial b_1 = x_1 +\nu\cdot x_1$; see Figure~\ref{fig:band}). 
  
Let $f\colon |\Gamma_L \times \Gamma_{L'}| \to Y$ be an equivariant map. By the (equivariant) cellular approximation theorem, $f$ induces an equivariant cochain map $f^*\colon C^\bullet(Y) \to C^\bullet (\Gamma_{L}\times \Gamma_{L'})$ (i.e., equivariant homomorphisms $f^*\colon C^i(Y) \to C^i (X)$ that commute with the coboundary map). We define
  \[
    \deg_1(f) = \bigl(f^*(e^0)(x_1) + f^*(d^0)(b_1)\bigr) \bmod 2
  \]
\end{definition}
  \begin{figure}
    \centering
    \begin{tikzpicture}[
  baseline=(current bounding box.south),
  scale=.4,
  every node/.style = {font=\small}]
    \def\size{8}
    \definecolor{mygray}{gray}{0.85}

    \foreach \i in {0,  ..., \size}{
        \foreach \j in {0, ..., \size}{
            \node[fill = mygray, circle, scale=.25] (\i\j) at (\i, \j) {};
        }
    }

    \foreach \i [remember=\i as \li (initially 0)] in {1, ..., \size}{
        \foreach \j [remember=\j as \lj (initially 0)] in {1, ..., \size}{
            \draw[mygray, thin, dashed] (\li\lj) -- (\li\j);
            \draw[mygray, thin, dashed] (\li\lj) -- (\i\lj);
            \draw[mygray, thin, dashed] (8\lj) -- (8\j);
            }
            \draw[mygray, thin, dashed] (\li8) -- (\i 8);
            }
         
    \draw[thin] (00) rectangle (\size\size);
    \draw[blue, very thick] (00.center) -- (80.center);
    \draw[blue, thick] (04.center) -- (84.center);
    \fill[blue, opacity = .1] (00.center) -- (80.center) -- (84.center) -- (04.center) -- (00.center);

    \node[blue] at (4, 0.3) {$x_1$};
    \node[blue] at (4, 4.3) {$\gromega \cdot x_1$};
    \node[blue] at (4, 2) {$b_1$};
\end{tikzpicture}
 \quad
    {\includegraphics{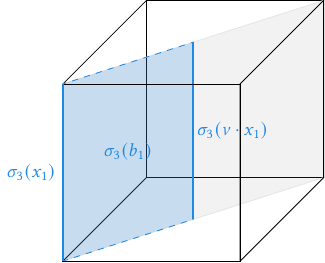}}
    \caption{Coordinate cycle and band in $T^2$ and $T^3$.}
    \label{fig:band} \label{fig:two_minor}
    \Description{A square whose lower half is highlighted and a cube with a diagonal rectangle with the half containing an edge of the cube highligted.}
  \end{figure}
  
Crucially, this notion of degree is invariant under equivariant homotopies:
\begin{lemma}\label{lem:2d_homotopy_invariance} 
  Fix positive integers $L, L_0$ and $L_1$ divisible by $4$. Let $f_0$ and $f_1 $ be equivariant maps $f_0 \colon |\Gamma_{L}\times \Gamma_{L_0}| \to \eml$ and $f_1 \colon | \Gamma_{L}\times \Gamma_{L_1}| \to \eml$ that are equivariantly homotopic. Then $\deg_1(f_0) = \deg_1(f_1)$.
\end{lemma}

\begin{proof} Assume first that $L_0 = L_1$. By the cellular approximation theorem again, there is an equivariant cochain homotopy between the induced cochain maps $f_0^*,f^*_1\colon C^\bullet(Y) \to C^\bullet (X)$, i.e., there exist equivariant homomorphisms $h\colon C^i(Y) \to C^{i-1}(X)$ satisfying 
\[f_0^* + f_1^*=\delta h+ h\delta.\]

Therefore, on the cochains of dimension $1$ we have:
\[
f_0^*(e^0)(x_1) + f_1^*(e^0)(x_1) = (\delta h(e^0))(x_1) + (h\delta(e^0))(x_1) = 0 + h(d^0 + d^1)(x_1)
    \]

   where the second equality is obtained by using the fact that $\partial x_1 = 0$ and $\delta e^0 = d^0 + d^1$.

    On the $2$-cochains we have:
    \[
        f_0^*(d^0)(b_1) + f_1^*(d^0)(b_1) = (\delta h(d^0))(b_1) + (h\delta(d^0))(b_1) = h(d^0)(x_1 + \gromega \cdot x_1) + 0
    \]
    where we use that $\partial b_1 = x_1 + \gromega \cdot x_1$ and $\delta d^0 = 0$.

    Moreover, $h$ is equivariant, hence $h(d^0)( \gromega \cdot x_1) = h(\gromega \cdot d^0)(x_1)$.
    Summing everything together, using this fact and that $\gromega \cdot d^0 = d^1$, we obtain that
    \begin{align*}
        \deg_1(f_0) + \deg_1(f_1) &= \bigl(f^*_0(e^0)(x_1) + f^*_1(e^0)(x_1)\bigr) + \bigl(f^*_0(d^0)(b_1) + f^*_1(d^0)(b_1)\bigr) \\
        &= h(d^0 + d^1)(x_1) +h(d^0)(x_1 + \gromega\cdot x_1) \\
        &= h(d^0)(x_1 + \gromega\cdot x_1 + x_1 + \gromega\cdot x_1) \\
        &= h(d^0)(2x_1) = 2h(d^0)(x_1) = 0 \pmod 2.
    \end{align*}
    
    If $L_0\neq L_1$, suppose without loss of generality that $L_0 < L_1$. Then $\Gamma_{L}\times \Gamma_{L_1}$ is a subdivision of $\Gamma_{L}\times \Gamma_{L_0}$, and the equivariant chain map $\iota \colon C_\bullet(\Gamma_{L}\times \Gamma_{L_0}) \to C_\bullet(\Gamma_{L}\times \Gamma_{L_1})$ that maps every $i$-cell $\sigma$ of $\Gamma_{L}\times \Gamma_{L_0}$ to the sum of $i$-cells of $\Gamma_{L}\times \Gamma_{L_1}$ that are contained in $\sigma$ is a chain homotopy equivalence.
Thus, by the previous case $\deg_1(f_0) = \deg_1(f_1\circ \iota)$ and from the definition of degree, $\deg_1(f_1\circ \iota) = \deg_1(f_1)$.
\end{proof}

We can now define $\deg_i(f)$ of an equivariant map $f\colon T^n \to Y$ as $\deg_1(f^{\sigma_i})$ for a suitable $2$-minor $f^{\sigma_i} \colon T^2 \to Y$ (see Figure~\ref{fig:two_minor}):

  \begin{definition}\label{def:degree}
    Let $L$ a positive integer divisible by $4$, and let $f\colon|\Gamma_L^n| \to \eml$ be a $\Ztwo$-equivariant map. 
    For $i \in [n]$, we define $\sigma_i\colon [n] \to [2]$ by $\sigma_i(i) = 1$ and $\sigma_i(j) = 2$ for $j \neq i$. Then the $i$-degree of $f$ is defined as
    \[
 \deg_i(f) = \deg_1(f^{\sigma_i}) = \bigl((f\circ \sigma_i)^*(e^0)(x_1) + (f\circ\sigma_i)^*(d^0)(b_1)\bigr) \pmod 2
    \]
 \end{definition}

An immediate consequence of Lemma~\ref{lem:2d_homotopy_invariance} is the invariance of the $i$-degree under equivariant homotopies:

\begin{corollary}\label{cor:invariance_degree}
  Let $f_0, f_1\colon|\Gamma^n_L| \to \eml$ be equivariant maps that are equivariantly homotopic. Then  $\deg_i(f_0) = \deg_i(f_1)$ for all $i\in [n]$. 
\end{corollary}

\begin{proof}
  Since $f_0$ and $f_1$ are equivariantly homotopic, so are their minors $f^{\sigma_i}_0$ and $f^{\sigma_i}_1$.
\end{proof}

It will be convenient to extend the notation for monomial maps to general integer coefficients. As before, let us view $S^1 =\{z \in \mathbb C \colon \lvert z\rvert=1\}$ as the unit circle in the complex plane. Given an $n$-tuple of integers $\alpha=(\alpha_1,\dots,\alpha_n)$ with $\sum_i \alpha_i \equiv 1 \bmod 2$, we get an equivariant map from $T^n$ to $S^1$ defined by $(z_1,\dots,z_n)\mapsto z_1^{\alpha_1}\cdots z_n^{\alpha_n}$. 
By composing this map first with a fixed equivariant inclusion $S^1 \hookrightarrow S^2$ and then with the inclusion $j\colon S^2 \rightarrow \eml$, we get an equivariant \emph{monomial map} $m_\alpha\colon T^n \to \eml$ given by 
 \[
    m_\alpha(z_1, \dots, z_n) =   j(z_1^{\alpha_1}\cdots z_n^{\alpha_n})
 \]
 
 \begin{remark}
It is not hard to observe that the assignment $\alpha \mapsto m_\alpha$ preserves minors when $\alpha$ is interpreted as a function $f\colon \Z^n \to \Z$. While we implicitly use this minion homomorphism, this is not the minion homomorphism we are looking for --- importantly, $\affine_2$ is \emph{not} a subminion of the minion of tuples $\alpha \in \Z^n$ with $\sum \alpha_i \equiv 1 \pmod 2$ since, e.g., the unary minor of $(1, 1, 1)$ disagrees in the two minions.
\end{remark}
 
Since monomial maps form a minion, we can easily compute the degree of any of them.
 
\begin{lemma}\label{lem:monomial_degrees}
    Let $\alpha \in \Z^n$ such that $\sum_i \alpha_i \equiv 1 \pmod 2$. Then
    \[
        \deg_i(m_\alpha) = \alpha_i \bmod 2
    \]
\end{lemma}

\begin{proof}
    Let $\sigma_i$ the minor used to define $\deg_i$. Then $m_\alpha^{\sigma_i} = m_{\beta}$ with $\beta = (\alpha_i, \sum_{j\neq i}\alpha_j)\in \Z^2$.
Since the image of $m_\beta$ is contained in the $1$-skeleton, $m_\beta^*(d^0)\equiv 0$.
    Moreover, $(m_\beta)_*(x_1) = \alpha_i e_0 +\alpha_i e_1$, hence $e^0((m_\beta)_*(x_1)) = \alpha_i$ and thus $\deg_i(m_\alpha) = \deg(m_\beta) = \alpha_i + 0\bmod 2$.
\end{proof}

\begin{corollary} \label{lem:B.3} \label{lem:all_distinct}
Let $\alpha, \beta \in \affine_2^{(n)}$. Then $m_\alpha$ and $m_\beta$ are equivariantly homotopic if and only if $\alpha = \beta$.
\end{corollary}

\begin{proof}If $\alpha=\beta$ then $m_\alpha$ and $m_\beta$ are identical as maps. Conversely, if $m_\alpha$ and $m_\beta$ are equivariantly homotopic, then $\deg_i(m_\alpha)=\deg_i(m_\beta)$ for all $i\in [n]$, by Corollary~\ref{cor:invariance_degree}. Thus, by  Lemma~\ref{lem:monomial_degrees}, $\alpha_i=\beta_i$, for all $i\in [n]$.
\end{proof}

We are now ready to prove Proposition~\ref{thm:affine_minion}:

\begin{proof}[Proof of Proposition~\ref{thm:affine_minion}]
For every $n\geq 1$, consider the map $\gamma_n\colon [T^n,Y]_{\Ztwo} \to \Z_2^n$ given by
  \[
   \gamma_n( [f]) = (\deg_1(f), \dots, \deg_n(f))
  \] 
By Corollary~\ref{cor:invariance_degree}, this mapping is well-defined. Moreover, by Lemma \ref{lem:monomial_degrees}, if $\alpha \in  \affine_2^{(n)}$, then the homotopy class $[m_\alpha] \in [T^n,Y]_{\Ztwo}$ of the corresponding monomial map satisfies $\gamma_n( [m_\alpha])=\alpha$, i.e., the homotopy classes $[m_\alpha]$, $\alpha \in  \affine_2^{(n)}$, are pairwise distinct, and by Lemma~\ref{lem:number-homotopy-classes}, they account for all elements of $ [T^n,Y]_{\Ztwo}$, i.e., every equivariant map $f\colon T^n\to Y$ is equivariantly homotopic to $m_\alpha$ for a unique $\alpha \in  \affine_2^{(n)}$. It follows that $\gamma_n( [f]) \in \affine_2^{(n)}$ and that $\gamma_n$ is a bijection. 

Furthermore, if $\alpha \in \Z^n$ with $\sum_i \alpha_i = 1$, and $\pi \colon [n] \to [m]$ then 
 \[
    \gamma_n([m_\alpha]) = (\alpha_1 \bmod 2, \dots, \alpha_n \bmod 2)
  \]
  by Lemma~\ref{lem:monomial_degrees}, hence $\gamma_n([m_\alpha])^\pi = \beta$, where $\beta_j = (\sum_{i\in \pi^{-1}(j)} \alpha_i)$\penalty200${}\bmod 2$.
  Furthermore, $m_\alpha^\pi = m_{\beta'}^{\phantom{\pi}}$ where
  $\beta'_j = \sum_{i\in \pi^{-1}(j)} \alpha_i$. Consequently,
  \[
    \gamma_m([m_\alpha]^\pi) = \beta = \gamma_m([m_\alpha^\pi]).
  \]
hence
\[
    \gamma_m([m_\alpha]^\pi) = \gamma_n([m_\alpha])^\pi,
\]
Thus, the maps $\gamma_n$ preserve minors for homotopy classes of monomial maps. Since these account for all homotopy classes, the maps $\gamma_n$ define a minion isomorphism $\gamma\colon \hpol{S^1, \eml} \to \affine_2$.
 \end{proof}

Finally, we show non zero degree guarantees a colour swapping edge.

\begin{lemma} \label{lem:col_swap_edges}
  \label{lem:2-torus-one-colour-swapping-edge}
  Let $f\colon \Gamma_L\times \Gamma_{L'} \to \Sigma^2$ a simplicial equivariant map such that $\deg_1(f) =1$. Then there is an horizontal color swapping edge, that is there is a vertex $(v_1, v_2)\in \Gamma_L\times \Gamma_{L'}$ such that $f(v_1, v_2) =\blueb$ and $f(v_1+1, v_2) = \yellowb$.
\end{lemma}

\begin{proof}
  Suppose, by contradiction, that every horizontal edge is monochrome. Therefore, the image of the horizontal coordinate cycle is constant so that $f^*(e_0)(x_1) = 0$. Additionally, the image of a triangle is non degenerate if and only if it is alternating (i.e., $f([u, v, w]) = [\blueb, \yellowb, \blueb]$ or $[\yellowb, \blueb, \yellowb]$); since we are assuming that every horizontal edge is monochrome, there are no alternating triangles and therefore $f^*(d^0)(b_1) = 0$. The total degree is then $\deg_1(f) = f^*(e_0)(x_1) +f^*(d^0)(b_1) = 0$
\end{proof}
 
\section{Bounding Essential Arity}
  \label{sec:bounded-arity}

We prove the key technical result that bounds the essential arity of simplicial maps from $\Gamma_L^n$ to $\Sigma^2$.

\begin{theorem}
\label{thm:bounded-arity}
Let $L \geq 4$ be an integer divisible by $4$, let $f\colon \Gamma_L^n \to \Sigma^2$ be an equivariant simplicial map such that the composition with the map $\Sigma^2 \to \eml$ is equivariantly homotopic to the map given by the monomial $\prod_{i\in \abs I} z_i$; equivalently, $\deg_i(f)=1$ if and only if $i\in I$.
Then $\abs I \leq O(L^2)$.
\end{theorem}

We recall (Observation~\ref{obs:2-coloring}) that equivariant simplicial maps $f\colon \Gamma_L^n \to \Sigma^2$ correspond bijectively to $2$-colourings of the vertices of $\Gamma_L^n$ with the following two properties: The colouring is equivariant (i.e., every pair of antipodal vertices of $\Gamma_L^n$ have distinct colours), and no $3$-simplex $[\boldsymbol u_0, \boldsymbol u_1, \boldsymbol u_2, \boldsymbol u_3]$ is coloured with alternating colours. We will show that this is impossible if $\lvert I\rvert$ is large; more precisely, we will show that if $i \in I$, then there are many edges $[\boldsymbol u, \boldsymbol v]$ such that the colours of $\boldsymbol u$ and $\boldsymbol v$ are different and $\boldsymbol u$ and $\boldsymbol v$ differ only in the $i$th coordinate (note that the difference in this coordinate is 1 by the definition of $\Gamma_L^n$). This is then used to show that we need to have an alternating simplex of dimension proportional to the size of $I$.

To present the details of the argument, we need a number of definitions. We recall the description of $\Gamma_L^n$: Its vertices are the $n$-tuples $\boldsymbol{u}=(u_1,\dots,u_n)\in \Z_L^n$;
edges ($1$-simplices) are pairs $[\boldsymbol{u},\boldsymbol{v}]$ of vertices such that $\boldsymbol{v}$ is obtained from $\boldsymbol{u}$ by choosing a non-empty subset of coordinates of $\boldsymbol{u}$ that are all even, and changing each of them by $\pm 1$ modulo $L$; and the $k$-simplices are $(k+1)$-tuples $[\boldsymbol{u}_0,\boldsymbol{u}_1,\dots, \boldsymbol{u}_k]$ such that $[\boldsymbol{u}_{j-1},\boldsymbol{u}_j]$ is an edge for $1\leq j\leq k$. We define the \emph{height} $\height(\boldsymbol u)$ of a vertex $\boldsymbol u = (u_1, \dots, u_n)$ as the number  of coordinates $i \in [n]$ such that $u_i$ is odd; moreover, we define the height of an edge $[\boldsymbol u, \boldsymbol v]$ as the height of $\boldsymbol u$. Note that for \emph{every} edge $[\boldsymbol u, \boldsymbol v]$, we have $\height(\boldsymbol u) < \height(\boldsymbol v)$. A special role will be played by edges $[\boldsymbol u, \boldsymbol v]$ such that $\height(\boldsymbol v) = \height(\boldsymbol u) + 1$, or equivalently, such that $\boldsymbol u$ and $\boldsymbol v$ differ in exactly one coordinate; we call such edges \emph{coordinate edges}. More precisely, we say that an edge $[\boldsymbol u, \boldsymbol v]$ is \emph{in coordinate direction $i$} if $\boldsymbol u$ and $\boldsymbol v$ differ exactly in the $i$th coordinate. For $i\in [n]$, we denote the set of all edges in coordinate direction $i$ by $E_i$, and denote by $E := E_1\sqcup \dots \sqcup E_n$ the set of all coordinate edges. We will also need the following more refined classification: For $i\in [n]$ and $0\leq h\leq n-1$, let $E_i(h)$ denote the set of all edges in $E_i$ of height $h$, and let $E(h) = E_1(h) \sqcup \dots \sqcup E_n(h)$ denote the set of all coordinate edges of height $h$ (note that the height $h$ of a coordinate edge determines the heights $h$ and $h+1$ of both endpoints).

Given a $2$-colouring of the vertices of $\Gamma_L^n$, we say that edge $[\boldsymbol u, \boldsymbol v]$ of $\Gamma_L^n$ is \emph{colour-swapping} if $\boldsymbol u$ and $\boldsymbol v$ have different colours.
We now state a key lemma used in the proof of Theorem~\ref{thm:bounded-arity}.
The lemma shows that, if $f$ depends on the coordinate $i$ (up to homotopy), then some fraction (independent from the arity of $f$) of edges in coordinate direction $i$ is colour-swapping.

\begin{lemma}
\label{lem:colour-swapping-by-height}
Let $f\colon \Gamma_L^n \to \Sigma^2$ be an equivariant simplicial map such that $\deg_i(f)=1$ and let $0\leq h < \lfloor \frac{n-1}{3} \rfloor$. Then a fraction of at least $\frac{1}{CL^2}$ of the edges in $E_i(h)\sqcup E_{i}(n-1-h)$ are colour-swapping, where $C>0$ is a suitable constant.
\end{lemma}

We postpone the proof of the lemma, and first show how it implies Theorem~\ref{thm:bounded-arity}.

\begin{proof}[Proof of Theorem~\ref{thm:bounded-arity} assuming Lemma~\ref{lem:colour-swapping-by-height}]
We first observe that the theorem reduces to the case that $n$ is odd and $I=[n]$. To see this, let $m = \abs{I}$, and choose any function $\pi \colon [n] \to [m]$ that is injective on $I$. Then the minor $f^\pi$ is an equivariant simplicial map $f^\pi \colon \Gamma_L^m \to \Sigma^2$ that is equivariantly homotopic to the monomial map $\prod_{i\in [m]} z_i$, by Lemma~\ref{lem:eta}.

Thus (by replacing $f$ by $f^\pi$ and $n$ by $m$), we may assume without loss of generality that $n$ is odd and $I=[n]$, i.e., $\deg_i(f) = 1$ for all $i \in [n]$. Now, consider a non-degenerate $n$-simplex $\sigma=[\boldsymbol u_0, \dots, \boldsymbol u_n]$ of $\Gamma_L^n$ chosen uniformly at random among all such $n$-simplices of $\Gamma_L^n$. For $0\leq h\leq n-1$, define the random variable $X_h(\sigma)$ as $1$ or $0$ depending on whether the edge $[\boldsymbol{u}_{h}, \boldsymbol{u}_{h+1}]$ is colour-swapping or not. Then $X(\sigma) := \sum_{h=0}^{n-1} X_h(\sigma)$ equals the total number of times the colour of $f(\boldsymbol{u}_{i})$ changes as we traverse the vertices of $\sigma$ in their given order.
  Observe that, for every $0\leq h\leq n-1$, the edge $[\boldsymbol u_{h}, \boldsymbol u_{h+1}]$ of the random simplex $\sigma$ is distributed uniformly among all edges of $E(h)$ (this is since the simplicial automorphisms of $\Gamma_L^n$ act transitively on $E(h)$).
  Thus, the expected value $\mathbf{E} [X_h(\sigma)]$ is the probability that a uniformly random edge in $E(h)$ is colour-swapping. Moreover, by Lemma~\ref{lem:colour-swapping-by-height} and summing over $1\leq i\leq n$, we get that for every $0\leq h < \lfloor \frac{n-1}{3} \rfloor$, the fraction of edges in $E(h)\sqcup E(n-1-h)$ that are colour-swapping is at least $\frac{1}{CL^2}$. Hence, by linearity of expectation,
  \(
    \mathbf{E} [X_h(\sigma)]+\mathbf{E} [X_{n-1-h}(\sigma)]\geq \frac{1}{CL^2}
  \)
  for $0\leq h < \lfloor \frac{n-1}{3} \rfloor$. Consequently,
  \[
    \mathbf{E}[X(\sigma)]=\sum_{h=0}^{n-1} \mathbf{E}[X_{n-1-h}(\sigma)]\ \geq \lfloor \frac{n-1}{3}\rfloor \cdot \frac{1}{CL^2}.
  \]
  Thus, there exists some $n$-simplex $\sigma=[\boldsymbol u_0, \dots, \boldsymbol u_n]$ of $\Gamma_L^n$ such that the colour of $f(\boldsymbol{u}_i)$ changes at least $k$ times, where $k = \lfloor \frac{n-1}3 \rfloor \cdot\frac1{CL^2}$, i.e., $\sigma$ contains some $k$-simplex $[\boldsymbol{u}_{i_0}, \boldsymbol{u}_{i_1}, \dots, \boldsymbol{u}_{i_k}]$ whose colours alternate.
  Since $f$ is a simplicial map to $\Sigma^2$, this implies that $k\leq 2$ as noted above, and therefore $\lfloor \frac{n-1}{3}\rfloor \leq 2CL^2$, hence $|I|=n=O(L^2)$.
\end{proof}

The rest of this section is dedicated to proving Lemma~\ref{lem:colour-swapping-by-height}.

In this proof, we will use Lemma~\ref{lem:2-torus-one-colour-swapping-edge} in combination with another averaging argument over a special family of triangulated $2$-dimensional tori $\Gamma_L \times \Gamma_{L'}$, which we call \emph{slices}, that are simplicially (and equivariantly) embedded in the triangulation $\Gamma_L^n$.

To simplify notation, let us fix a coordinate direction, say $i=1$, and write $\Gamma_L \times \Gamma_L^{n-1}$. The archetype of a slice is the following \emph{standard slice}: Consider the \emph{diagonal embedding} $\diag\colon \Gamma_L \hookrightarrow \Gamma_L^{n-1}$ given by $\diag(y) = (y,\dots,y)$. This is an equivariant simplicial map, which induces an equivariant simplicial embedding $\slice_{\diag} \colon \Gamma_L \times \Gamma_L \hookrightarrow \Gamma_L^n$ given by $\slice_{\diag} \coloneq 1_{\Gamma_L}\times \diag$, i.e., $\slice_{\diag}(x,y)=(x,y,\dots,y)$.

More generally, let $L'$ be an integer divisible by $4$, and let $\zeta \colon \Gamma_{L'} \to \Gamma_L^{n-1}$ be an equivariant simplicial map; we call $\zeta$ a \emph{generalized diagonal} if its geometric realization $\geom{\zeta}$, seen as an equivariant embedding  $S^1 \to T^{n-1}$, is equivariantly homotopic to the diagonal embedding $S^1 \to T^{n-1}$ (here, we implicitly fix equivariant homeomorphisms $\geom{\Gamma_{L'}} \cong S^1 \cong \geom{\Gamma_L}$).
Given a generalized diagonal $\zeta$, we call the induced equivariant simplicial embedding $\slice_{\zeta} \colon \Gamma_L \times \Gamma_{L'} \to \Gamma_L^n$ given by $\slice_{\zeta} = 1_{\Gamma_L} \times \zeta$ a~\emph{slice}.
Moreover, we say that $\slice_\zeta$ is an \emph{$h$-slice} if every vertex of $\Gamma_L^{n-1}$ in the image of $\zeta$ is at height $h$ or $n-1-h$, or equivalently, if every edge of $\Gamma_L^n$ that lies in both $E_1$ and the image of $\slice_{\zeta}$ belongs to $E_1(h)\sqcup E_1(n-1-h)$.

\begin{lemma}
\label{lem:colour-swapping-slice}
Let $f\colon \Gamma_L^n \to \Sigma^2$ be an equivariant simplicial map such that $\deg_1(f)=1$, and let $\slice_{\zeta}\colon \Gamma_L\times \Gamma_{L'}\to \Gamma_L^n$ be a slice (respectively, an $h$-slice, $0\leq h\leq n-1$). Then the image of $\slice_{\zeta}$ contains at least one edge in $E_1$ (respectively, in $E_1(h)\sqcup E_1(n-1-h)$) that is colour-swapping.
\end{lemma}
\begin{proof}
The composition $f\circ \slice_{\diag}$ is the same as the $2$-minor $f^\pi$ of $f$ given by the map $\pi\colon [n]\to [2]$, $\pi(1)=1$ and $\pi(j)=2$ for $2\leq j \leq n$. Thus, $\deg_1(f\circ \slice_{\diag})=\deg_1(f^\pi)=\deg_1(f)=1$ by Definition~\ref{def:degree}. Moreover, by definition of generalized diagonals, it follows that $\geom{f\circ \slice_{\zeta}}$ and $\geom{f\circ \slice_{\diag}}$ are equivariantly homotopic as maps $T^2=S^1 \times S^1 \to S^2$, hence $\deg_1(f\circ \slice_{\zeta})=\deg_1(f\circ \slice_{\diag})=1$ (here, we use that the equivariant homeomorphism $\geom{\Gamma_L \times \Gamma_{L'}} \cong \geom{\Gamma_L \times \Gamma_L}$ fixes the two coordinate copies of $S^1$ in $T^2$). Thus, the existence of the desired colour-swapping edge follows from Lemma~\ref{lem:2-torus-one-colour-swapping-edge}.
\end{proof}

The last puzzle piece we need to prove Lemma~\ref{lem:colour-swapping-by-height} (and thus to complete the proof of Theorem~\ref{thm:bounded-arity}) is the following lemma which constructs a generalised diagonal of a special shape.

\begin{figure}
  \[
    \includegraphics[scale=.75]{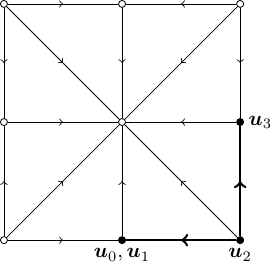}
    \quad
    \includegraphics[scale=.75]{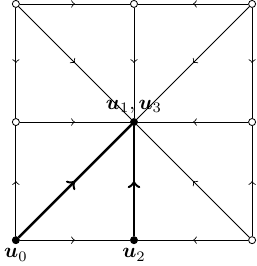}
  \]
  \caption{A path starting with the point $\boldsymbol u_0 = (1, 0, 0, 0)$ shown as projection on the first two (left) and last two coordinates (right).}
  \label{fig:slice}
  \Description{A visualisation of the path from $\boldsymbol u_0$ to $\boldsymbol u_3$; described precisely in the text.}
\end{figure}

\begin{lemma}\label{lem:slices}
  Let $0\leq h < \lfloor \frac{n-1}{3} \rfloor$. Then there exists a generalised diagonal $\zeta_0\colon \Gamma_{3L} \to \Gamma_L^{n-1}$ whose image contains only vertices of height $h$ or $n-1-h$; moreover, the vertices of height $h$ and $n-1-h$ alternate.
\end{lemma}

\begin{proof}
  We start with constructing a simplicial map $\zeta_0\colon \Gamma_{3L} \to \Gamma_L^{n-1}$, i.e., a cyclic path in $\Gamma_L^{n-1}$, that contains only vertices of height $h$ or $n-1-h$.

  We start with the vertex $\boldsymbol u_0$ of the form $\boldsymbol u_0 = (1, \dots, 1, 0, \dots, 0)$ where the first $h$ coordinates are $1$, and construct a path from $\boldsymbol u_0$ to its antipode in pieces of length 3.
  The first three steps of the path have the following form (where the first three blocks are of length $h$ and the last block is of length $n-1 - 3h$).
  \begin{align*}
    \boldsymbol{u}_0 &=
      (\underbrace{1,\dots,1}_{h},\underbrace{0,\dots,0, 0,\dots,0, 0,\dots,0}_{n-1-h}) \\
    \boldsymbol{u}_1 &=
      (\underbrace{1,\dots,1}_{h},\underbrace{0,\dots,0}_{h},\underbrace{1,\dots,1, 1,\dots,1}_{n-1-2h}) \\
    \boldsymbol{u}_2 &=
      (\underbrace{2,\dots,2}_{h},\underbrace{0,\dots,0}_{h},\underbrace{1,\dots,1}_{h},\underbrace{0,\dots,0}_{n-1-3h})  \\
    \boldsymbol{u}_3 &=
      (\underbrace{2,\dots,2}_{h},\underbrace{1,\dots,1, 1,\dots,1, 1,\dots,1}_{n-1-h})
  \end{align*}
  In the first step, we increase the values in the last two blocks changing $h+ (n-1-3h) = n-1-2h$ values. In the second step, we increase the value in the first bloc and decrease the value in the last bloc, again changing the same number of values. And in the third step, we increase the values in the second and the last block. See also Fig.~\ref{fig:slice} for a visual representation of the case $n = 4$ and $h = 1$.
  Note that the height of $\boldsymbol u_0$ and $\boldsymbol u_2$ is $h$ and the height of $\boldsymbol u_1$ and $\boldsymbol u_3$ is $n-1-h$, and that $\boldsymbol u_3$ is $\boldsymbol u_0$ shifted along the diagonal by $1$.

  We then repeat this pattern (until we return to $\boldsymbol u_0$) by adding $1$ to all coordinates in each subsequent sequence of three steps, i.e.,
  \[
    \boldsymbol u_4 = (\underbrace{2,\dots,2}_{h},\underbrace{1,\dots,1}_{h},\underbrace{2,\dots,2, 2,\dots,2}_{n-1-2h}),
  \]
  etc. It is easy to check that the height of $\boldsymbol u_{2k}$ is $h$ and the height of $\boldsymbol u_{2k+1}$ is $n-1-h$ for all $k$, and that subsequent vertices are connected by an edge in $\Gamma_L^{n-1}$. Furthermore, observe that  $\boldsymbol u_{k+3L/2} = \boldsymbol u_k + \frac L2 \mathbbm 1$ is the antipode of $\boldsymbol u_k$, hence $\zeta_0\colon \Gamma_{3L} \to \Gamma_L^{n-1}$ defined by $\zeta_0(k) = \boldsymbol u_k$ is an equivariant simplicial map.

  Next, we prove that $\zeta_0$ is a generalized diagonal.
  We view the geometric realization of $\Gamma_L$ as $\mathbb R / L \mathbb Z \cong S^1$.
  Observe that every point $\boldsymbol{x} = (x_1,\dots,x_{n-1}) \in T^{n-1}$ on the (geometric realization of the) path from $\boldsymbol{u}_0$ to $\boldsymbol u_3 = \boldsymbol u_0+\mathbbm{1}$ satisfies $x_i \in [1, 2]$ if $i \leq h$ and $x_i \in [0, 1]$ if $i > h$; thus, $\boldsymbol x \in [1, 2]^h\times [0, 1]^{n-1-h}$, i.e., $\boldsymbol x$ lies inside a unit box.
  Since this box is convex, we can homotope the path to the ``straight'' path from $\boldsymbol{u}$ to $\boldsymbol{u}+\mathbbm{1}$ inside the box, keeping the endpoints fixed, by linear interpolation.
  By an analogous argument applied to each path segment corresponding to a sequence of three steps from $\boldsymbol u_{3k}$ to $\boldsymbol u_{3k+3}$, we get a homotopy between the embedding $\zeta_0$ and a translated copy of the diagonal that passes through $\boldsymbol{u}_0$. Moreover, this translated copy to the diagonal is homotopic to the diagonal itself, hence $\zeta_0$ is a generalized diagonal (note that translated copies of the diagonal are not simplicial embeddings in general, which is why we use the more complicated construction).
\end{proof}

We may now finish the proof of Lemma~\ref{lem:colour-swapping-by-height} and, consequently, of Theorem~\ref{thm:bounded-arity}.

\begin{proof}[Proof of Lemma~\ref{lem:colour-swapping-by-height}]
  Let us fix a coordinate direction, without loss of generality $i=1$, and let $f\colon \Gamma_L^n \to \Sigma^2$ be an equivariant simplicial map such that $\deg_1(f)=1$. Let $0\leq h < \lfloor \frac{n-1}{3} \rfloor$.

  First, we prove that there exists a collection $Z$ of generalised diagonals that contain only vertices of heights $h$ and $n - h - 1$ such that each vertex of such a height appears in the same number of diagonals accounting for multiplicity. This collection is constructed by shifting the diagonal $\zeta_0$ obtained in Lemma~\ref{lem:slices} by some automorphisms of $\Gamma_L^{n-1}$. We consider only those automorphisms that respect the winding direction in each coordinate, which consequently respect also the homotopy class of the diagonal. More precisely, consider the subgroup $A$ of automorphism group of $\Gamma_L^{n-1}$ generated by automorphisms of one of the following two types:
  \begin{itemize}
    \item $a_\pi$, where $\pi \colon [{n-1}] \to [{n-1}]$ is permutation, which permutes the coordinates of each vertex, i.e.,
      \[
        a_\pi(u_1, \dots, u_{n-1}) = (u_{\pi(1)}, \dots, u_{\pi({n-1})});
      \]
    \item $b_i$, where $i\in [{n-1}]$, which shifts the coordinate $i$ by $2$, i.e.,
      \[
        b_i(u_1, \dots, u_{n-1}) =
        (u_1, \dots, u_{i-1}, (u_i + 2)\bmod L, u_{i+1}, \dots, u_{n-1}).
      \]
  \end{itemize}
  Observe that $A$ acts transitively on vertices of height $h$: for example, first use $b_i$'s to make all coordinates $0$ or $1$, and then use $a_\pi$ to permute them in the first $h$ positions. In fact, the orbits of $A$ are exactly sets of vertices of the same height.
  Now, we let $Z = \{g\circ \zeta_0 \mid g \in A\}$. Since this family is invariant under the action of $A$ which, as we said, is transitive on vertices of height $h$ and of height $n-h-1$, respectively, each such vertex appears in the same number of generalised diagonals in $Z$.
  Since the vertices of height $h$ and $n-1-h$ alternate in $\zeta_0$, and consequently, they alternate in each of the shifts, we also get the number of times a vertex of height $h$ appears is the same as the number of times a vertex of height $n-1-h$ appears.

  For each $\zeta \in Z$, the image of the corresponding $h$-slice $\slice_{\zeta}\colon \Gamma_L \times \Gamma_{3L} \to \Gamma_L^n$ contains $3L^2$ edges in $E_1(h)\sqcup E_1(n-1-h)$, and at least one of these edges is colour-swapping, by Lemma~\ref{lem:colour-swapping-slice}. Moreover, the number of slices $\zeta \in Z$ whose image contain a given edge in $E_i(h)\sqcup E_i(n-1-h)$ does not depend on the edge.
  Thus, we can choose a uniformly random element of $E_i(h)\sqcup E_i(n-1-h)$ by first choosing a uniformly random element $\zeta\in Z$, and then choosing uniformly at random a vertex $v\in \Gamma_{3L}$ and an edge in coordinate direction $i$ which projects to $\zeta(v)$. Since the probability that we selected a colour-swapping edge in the last choice is at least $\frac 1{3L^2}$, the overall probability that an uniformly random edge from $E_i(h)\sqcup E_i(n-1-h)$ is colour-swapping is also at least $\frac{1}{3L^2}$.
\end{proof}

\section{Minion homomorphisms}
  \label{app:minion-homomorphism}

In this section, we construct the minion homomorphisms that we use in the proof of Theorem~\ref{thm:main-cycles} and we prove Lemma~\ref{lem:minion-homomorphism}.

Fix an odd integer $\ell\geq 3$. We describe three maps $\mu$, $\eta$, and $\phi$ between the minions $\pol(C_\ell, K_4)$, $\spol(\Gamma_{4\ell}, \Sigma^2)$, and $\hpol{S^1, \eml}$; diagrammatically, these maps are organised as follows:
\[\begin{tikzcd}
  \pol(C_\ell, K_4) \arrow[r, dashrightarrow, "\mu"] \arrow[rd, "\phi"']
                & \spol(\Gamma_{4\ell}, \Sigma^2) \arrow[d, "\eta"]\\
                & \hpol{S^1, \eml}
\end{tikzcd}\]
The diagram commutes, i.e., $\phi = \eta\circ\mu$. Furthermore, $\eta$ and $\phi$ are minion homomorphisms, while $\mu$ preserves minors only up to homotopy.

On a high level, all of these maps are constructed using functors that preserve products, or preserve products up to homotopy equivalence. Detailed proofs of the technical results we present here can be found in the arXiv version of~\cite[Section~4]{FilakovskyNOTW26}. For the reader's convenience, we sketch these proofs here and refer to \cite{FilakovskyNOTW26} for the details. We will use the following lemma.

\begin{lemma}[Relaxation lemma {\cite[Lemma 4.16]{FilakovskyNOTW26}}]
  \label{lem:relaxation}
  Let $\ssX$, $\ssY$, $\ssX'$, and $\ssY'$ be simplicial sets with $\Ztwo$-actions such that there are equivariant simplicial maps $\ssX' \to \ssX$ and $\ssY \to \ssY'$. Then there is a minion homomorphism $\spol(\ssX, \ssY) \to \spol(\ssX', \ssY')$.
  
  The same is true for topological spaces in place of simplicial sets, continuous maps in place of simplicial maps, and $\operatorname{hpol}$ in place of $\spol$.
\end{lemma}

\begin{proof}[Proof sketch]
  Given the two simplicial maps $t\colon \ssX'\to \ssX$ and $s\colon \ssY\to \ssY'$, the minion homomorphism obtained by mapping a map $f$ of arity $n$ to the composition $sf(t(x_1), \dots, t(x_{n}))$. It is easy to check that this indeed preserves all minors.
\end{proof}

\subsection{From graphs to simplicial sets}

In essence, the fact that $\mu$ is a minion homomorphism follows from the fact that the homomorphism complex of a product of two graphs is equivariantly homotopy equivalent to the product of homomorphism complexes (see, e.g., \cite[Proposition 18.17]{Koz08}), and by applying Lemma~\ref{lem:relaxation} to the equivariant simplicial map $s\colon \Hom(K_2, K_4) \to \Sigma^2$ described in Lemma~\ref{lem:simplicial-map-HomK4-Sigma2}.
We construct $\mu$ in two steps: First we will go from graphs to multihomomorphism posets, and then from these posets to simplicial sets.

Let $P, Q$ be posets. By definition, an $n$-ary \emph{poset polymorphism} from $P$ to $Q$ is a map $f\colon P^n \to Q$ that is monotone (where the partial order on $P$ is defined componentwise). We use the usual notation $\pol^{(n)}(P, Q)$ and $\pol(P, Q)$ and the sets of polymorphisms.

Monotone maps between posets are naturally partially ordered: $f \leq g$ if $f(x) \leq g(x)$ for all $x$. This allows us to relax the notion of minion homomorphism: Let $\minion M$ be a minion, and $P, Q$ posets. A \emph{lax minion homomorphism} $\minion M \to \pol(P, Q)$ is a collection of mappings $\lambda_n\colon \minion M^{(n)} \to \pol^{(n)}(P, Q)$ such that $\lambda_m(f^\pi) \leq \lambda_n(f)^\pi$. The following is a straightforward generalisation of \cite[Lemma 4.1]{MeyerO24}.

\begin{lemma} \label{lem:lax}
  Let $G, H$ be graphs. There is a lax minion homomorphism
  \[
    \mu' \colon \pol(G, H) \to \pol(\mhom(K_2, G), \mhom(K_2, H)).
  \]
\end{lemma}

\begin{proof}
  Let $f\colon G^n\to H$ be a homomorphism. We define
  \[
    \mu'(f)\colon \mhom(K_2, G)^n \to \mhom(K_2, H)
  \]
  by setting $\mu'(f)(m_1, \dots, m_n)$ to be the multihomomorphism
  \[
    u \mapsto \{ f(v_1, \dots, v_n) \mid v_i \in m_i(u) \text{ for } i \in [n] \}
  \]
  where $u \in V(C)$. It is easy to check that $\mu'(f)(m_1, \dots, m_n)$ is a multihomomorphism using that all $m_i$'s are multihomomorphisms and $f$ is a polymorphism.
  
Now consider a map $\pi\colon [n] \to [k]$ and multihomomorphisms $m_j \in \mhom(C, G)$ for $j\in [k]$. For every vertex $u$ of $K_2$, we have
  \begin{align*}
    \mu'(f)^\pi(m_1, \dots, m_k)(u) &= \mu'(f)(m_{\pi(1)}, \dots, m_{\pi(n)})(u) \\
             &= \{f(v_1, \dots, v_{n}) \mid v_i \in m_{\pi(i)}(u) \textrm{ for all } i\in[n] \} \\
             &\supseteq \{f(v'_{\pi(1)}, \dots, v'_{\pi(n)}) \mid v'_j \in m_j(u) \textrm{ for all } j\in[k]\} \\
             &= \{ f^\pi(v'_1, \dots, v'_k) \mid v'_j \in m_j(u) \textrm{ for all } j\in [k] \} \\
             &= \mu'(f^\pi)(m_1, \dots, m_k)(u)
  \end{align*}
Thus, $\mu'(f)^\pi \geq \mu'(f^\pi)$ as we wanted to show. Checking that $\mu'(f)$ preserves the $\Ztwo$-symmetry is straightforward.
\end{proof}

By applying the monotone map $\mu'(f)$ elementwise to chains, it naturally extends to a simplicial map 
\[
  \mu'(f)\colon \Hom(K_2, G)^n \to \Hom(K_2,H).
\]
In this way, $\mu'$ can be treated as a map 
\[
  \mu'\colon \pol(G, H) \to \spol(\Hom(C, G), \Hom(C, H)).
\]
In order to show that $\mu'$ preserves minors up to homotopy, we use the following well-known result about order complexes (see, e.g., \cite[Theorem~10.11]{Bjorner:1995} or \cite[Lemma 2.3]{MeyerO24}):

\begin{lemma} \label{lem:homotopy-on-equivariant-posets}
If $f, g\colon P\to Q$ monotone are monotone maps between posets such that $f\geq g$, then the induced continuous maps $\geom f, \geom g\colon \geom{ \Delta(P)}\to \geom{\Delta(Q)}$ are homotopic. Moreover, if $\Ztwo$ acts on both $P$ and $Q$ and $f$ and $g$ are equivariant, then $\geom f$ and $ \geom g$ are equivariantly homotopic.
\end{lemma}

\begin{proof}
  Consider the poset $P \times \{0, 1\}$ with the componentwise partial order, where $\Ztwo$ acts trivially on the first coordinate. Since $f\geq g$, the map $H\colon P\times \{0, 1\}$ defined by $H(p, 0) = f(p)$ and $H(p, 1) = g(p)$ is monotone and equivariant. Further observe that $\geom {\Delta(P \times \{0, 1\})}$ is $\Ztwo$-homeomorphic to $\geom{\Delta(P)} \times [0, 1]$, hence $\geom H$ induces an equivariant homotopy $\geom{\Delta(P)} \times [0, 1] \to \geom{\Delta(Q)}$ with $\geom H(-, 0) = \geom f$ and $\geom H(-, 1) = \geom g$.
\end{proof}

Using Lemma~\ref{lem:lax} (applied with $G = C_\ell$ and $H = K_4$), Lemma~\ref{lem:homotopy-on-equivariant-posets}, and the equivariant simplicial map $s\colon \Hom(K_2, K_4) \to \Sigma^2$ described in Lemma~\ref{lem:simplicial-map-HomK4-Sigma2}, we get the required homomorphism $\mu$:

\begin{lemma} \label{lem:mu}
  There is a mapping $\mu\colon \pol(C_\ell, K_4) \to \spol(\Gamma_{4\ell}, \Sigma^2)$ such that $\geom {\mu(f^\pi)}$ and $\geom {\mu(f)^\pi}$ are $\Ztwo$-homotopic for all polymorphisms $f\in \pol^{(n)}(C_\ell, K_4)$ and $\pi\colon [n] \to [m]$.
\end{lemma}

\begin{proof}
 Let $s\colon \Hom(K_2, K_4) \to \Sigma^2$ be the equivariant simplicial map described in Lemma~\ref{lem:simplicial-map-HomK4-Sigma2}. Then $\mu$ is defined by $\mu(f) = s\circ\mu'(f)$. We have $\mu'(f^\pi) \leq \mu'(f)^\pi$, hence the geometric realisations of these two maps are equivariantly  homotopic by Lemma~\ref{lem:homotopy-on-equivariant-posets}. Composing with $s$ preserves both 
minors and equivariant homotopies.
\end{proof}

\subsection{From simplicial sets to topological spaces}

The minion homomorphism $\eta$ is a composition of the minion homomorphism obtained by geometric realisation and the relaxation lemma (Lemma~\ref{lem:relaxation}). This has been discussed in detail in \cite[Section~4]{FilakovskyNOTW26}; let us outline the key ideas here.

Firstly, we use the fact that geometric realisation preserves finite products \cite[Theorem 5.12]{Friedman2012},\footnote{In our case, the simplicial sets are locally finite, hence the statement is true for the usual product of topological spaces.} and hence, for any simplicial map $f\colon \ssX^n \to \ssY$, we can treat $\geom f$ as a function $\geom \ssX^n \to \geom\ssY$. The following is then an instance of a more abstract statement \cite[Lemma 4.15]{FilakovskyNOTW26}.

\begin{lemma}
  Let $\ssX$ and $\ssY$ be two simplicial sets, then the mapping
  \[
    \eta' \colon \spol(\ssX, \ssY) \to \hpol{\geom \ssX, \geom \ssY}
  \]
  defined by $\eta'(f) = [\geom f]$ is a minion homomorphism.
\end{lemma}

\begin{proof}[Proof sketch]
  After we have identified $\geom {\ssX^n}$ with ${\geom \ssX}^n$, there is not much happening here. The functions $\geom f^\pi$ and $\geom{f^\pi}$ agree on vertices since on those they are both defined as $f^\pi$. Similarly, they map each of the faces to the same face. Finally, on internal points of the faces, they are both defined as a linear extension of $f^\pi$, and hence they are equal, and consequently, their homotopy classes coincide.
\end{proof}

Combining the above with the relaxation lemma using the $\Ztwo$-equivariant continuous map $S^2 \to \eml$ constructed in Lemma~\ref{lem:Postnikov} below, we obtain the desired minion homomorphism.

\begin{lemma} \label{lem:eta}
  There is a minion homomorphism
  \[
    \eta\colon \spol(\Gamma_{4\ell}, \Sigma^2) \to \hpol{S^1, \eml}.
  \]
\end{lemma}

\subsection{From graphs to topological spaces}

Finally, let us discuss the composition $\phi = \eta \circ \mu$. The claim is the following.

\begin{lemma} \label{lem:phi}
  The composition $\phi = \eta \circ \mu$ is a minion homomorphism
  \[
    \phi \colon \pol(C_{\ell}, K_4) \to \hpol{S^1, \eml}.
  \]
\end{lemma}

\begin{proof}
  It is enough to show that the composition preserves minors. To show that, let $f\in \pol^{(n)}(C_\ell, K_4)$ and $\pi \colon [n] \to [m]$. We have that $\mu(f)^\pi$ and $\mu(f^\pi)$ are $\Ztwo$-homotopic by Lemma~\ref{lem:mu}. Furthermore,
  \[
    \phi(f)^\pi = (\eta\mu(f))^\pi = \eta(\mu(f)^\pi) = \eta\mu(f^\pi) = \phi(f^\pi)
  \]
  where the third equality uses the fact that $\eta$ is constant on homotopy classes (which is true since $\eta$ is a composition of $\eta'$ and postcomposition with $S^2 \to \eml$, and $\eta'$ is constant on homotopy classes by definition).
\end{proof}

This concludes the proof of Lemma~\ref{lem:minion-homomorphism}.
 
\section{Equivariant topology}
\label{app:cohomology}
\label{app:cohomology_def}

In this section, we describe how to construct, starting from $S^2$, a $\Ztwo$-space $\eml$, that is homotopically simpler, together with an equivariant map $S^2\to \eml$. We also describe some properties of the space used in our proof, in particular, we prove Lemma~\ref{lem:number-homotopy-classes}.

On the high level, the space $\eml$ is constructed so that its homotopy groups $\pi_n(\eml)$ agree with the homotopy groups of $S^2$ up to dimension 2, i.e., for $n \leq 2$, and are trivial for $n > 2$ (the latter is not true for $S^2$). Note that this means that the only non-trivial homotopy group of $\eml$ is $\pi_2(\eml) = \pi_2(S^2) \cong \Z$. Thus, the space $\eml$ is an \emph{Eilenberg--MacLane space}; see \cite[p.~365]{Hat02}.
The homotopy classes of maps from a complex $X$ to an Eilenberg--MacLane space are in bijection with the elements of a suitable cohomology group of $X$ \cite[Theorem 4.57]{Hat02}. An analogous statement is also true in the equivariant setting; this will allow us to determine $[T^n, \eml]_\Ztwo$ by computing a suitable equivariant cohomology group, specifically the \emph{Bredon cohomology group} $H^2_\Ztwo(T^n;\pi_2(S^2))$ (see Definition~\ref{def:bredon-cohomology}), which will allow us to prove Lemma~\ref{lem:number-homotopy-classes}.

Throughout this section, we assume some familiarity with fundamental notions of algebraic topology such as homotopy, homology and cohomology. We refer to \citet{Hat02} for background on the more basic non-equivariant setting, and to May et al.~\cite[Chapters I and II]{MayCP96}, \citet{Die87}, and \citet{Bre67} for more details on equivariant homotopy and cohomology of spaces with group actions (all the definitions and constructions we use are special cases of the general theory described in these standard references).

If $X$ is a topological space with a $\Ztwo$-action (given by a continuous involution $\nu\colon X\to X)$, we will simply refer to $X$ as a \emph{$\Ztwo$-space}, and we use the multiplicative notation $\nu \cdot x$ instead $\nu(x)$.

\subsection{Construction of the space \texorpdfstring{$\eml$}{Y}}
\label{sec:construction-EML}

There are several different but ultimately equivalent ways of constructing the space $\eml$; here (following \citet[Example~4.13]{Hat02}), we will use a simple inductive construction  that starts with the sphere $S^2$ and achieves triviality of the higher homotopy groups $\pi_i(\eml)$, $i > 2$, by successively gluing the boundaries of higher and higher-dimensional disks along non-trivial elements of the corresponding homotopy group. The formal description of this construction uses the notion of \emph{CW complexes}.

A \emph{CW complex} is a space $X$ together with a increasing sequence of subspaces (called a \emph{filtration}) 
\[
  X_0 \subseteq X_1 \subseteq X_2 \subseteq \cdots\subseteq X,
\]
with the following properties: $X_0$ is a discrete set of points (called \emph{vertices} or \emph{$0$-dimensional cells}) and $X_{i+1}$ is constructed by attaching a set of $(i+1)$-dimensional discs $D^{i+1}_\alpha$ to $X_i$ along their boundary via continuous maps $g_\alpha \colon \partial D^{i+1} _\alpha = S^{i} _\alpha \to X_i$. Thus
\[
  X_{i+1} = \faktor{(X_i \sqcup \coprod_\alpha D^{i+1}_\alpha)}\sim
\]
where $\sim$ identifies $g_\alpha(x)\in X_i$ with $x\in \bd D^{i+1}_\alpha$. Finally, the topology on $X = \bigcup_n X_n$ is the so-called \emph{weak topology} (i.e., a set $U\subseteq X$ is open if and only if $X\cap X_i$ is open in $X_i$ for every $i$). The subspace $X_i$ is called the $i$-dimensional skeleton of $X$.

We say that $X$ is $\Ztwo$-\emph{CW complex} if, for each $i \geq 0$, $\Ztwo$ acts on the set of $i$-simplices and the attaching maps respect the action. As remarked above, the geometric realization $\geom{X}$ of simplicial set $X$ is a CW complexes, and if $X$ has a simplicial $\Ztwo$-action, then $\geom{X}$ is a $\Ztwo$-CW complex.

\begin{lemma}\label{lem:Postnikov}
  There exists a $\Ztwo$-CW complex $\eml$ such that
  \begin{enumerate}
    \item $\pi_2(\eml) = \pi_2(S^2)$;
    \item $\pi_i(\eml) = 0$ for all $i \neq 2$; and
    \item there is a $\Ztwo$-map $j \colon S^2 \to \eml$ that induces an isomorphism $\pi_2(j) \colon \pi_2(S^2) \to \pi_2(\eml)$ of groups with a $\Ztwo$-action.
  \end{enumerate}  
\end{lemma}

\begin{proof}[Proof sketch]
The sphere $S^2$ can be viewed as a $\Ztwo$-CW complex with the antipodal action (we can, e.g., take the geometric realization of the simplicial set $\Sigma^2$).
Starting with this $\Ztwo$-CW complex, we construct the $i$-skeleton $\eml_i$ of $\eml$ as follows:
\begin{enumerate}
  \item We set $\eml_2 \coloneq S^2$.
  \item For $i > 2$ we create a space $\eml_{i}$ as follows: Start with $\eml_{i-1}$ and for every generator $\alpha$ of $\pi_i(\eml_{i-1})$ we attach two $(i+1)$-dimensional discs $D_\alpha, D_{\gromega\cdot \alpha}$ by identifying $\partial D_\alpha$ with $\alpha$ and $\partial D_{\gromega\cdot \alpha}$ with $\gromega\cdot \alpha$. We then extend the $\Ztwo$ action in the natural way by ``swapping'' the paired discs $D_\alpha$ and  $D_{\gromega\cdot \alpha}$.
  \item Finally, we take $\eml = \bigcup_{i \geq 2} \eml_i$.
\end{enumerate}
It is not hard to check (see \cite[Example~4.13]{Hat02}) that the $\Ztwo$-CW complex $\eml$ satisfies $\pi_i(\eml) = 0$ for $i \geq 3$ and $\pi_i(\eml) = \pi_i(S^2)$ for $i \leq 2$.
Further, $j\colon S^2\to \eml$ is defined as the inclusion of the $2$-skeleton $\eml_2 = S^2$ into $\eml$.
\end{proof}

\subsection{Equivariant cohomology: A primer}
\label{sec:cohomology}

We now introduce the Bredon cohomology that will help us classify equivariant maps $T^n \to \eml$.

Prescribing a $\Z_2$-action on an Abelian group $M$ is the same as giving $M$ the structure of a module over the \emph{group ring} $\grLambda$ (which is isomorphic to the the quotient $\Z[\gromega]/(\gromega^2 - 1)$ of the polynomial ring by the ideal 
$(\gromega^2 - 1)$). In particular, if $Y$ is a space with a $\Ztwo$-action, then this action naturally induces a $\Ztwo$-action on every homotopy group $\pi_i(Y)$ and hence turns $\pi_i(Y)$ into a $\grLambda$-module. In what follows, we will mainly use the terminology of $\grLambda$-modules (rather than speaking of abelian groups with $\Ztwo$-actions). We are now ready to recall the definition of equivariant homology and cohomology groups.

\begin{definition}[Equivariant (co)homology]
  \label{def:bredon-cohomology}
  Let $X$ be a $\Ztwo$-CW complex. Its $d$-dimensional chain group $C_d(X)$ has a natural structure of $\grLambda$-module with multiplication given on a cell $\sigma$ by
  \[
    (n_0+n_1\gromega) \sigma =  n_0 \sigma + n_1 (\gromega\cdot \sigma)
  \]
  and extended linearly. Since the all the boundary maps commute with the action, these are $\grLambda$-module homomorphisms, and hence $C_\bullet(X)$ can be viewed as a chain complex of $\grLambda$-modules. We denote this chain complex by $C^{\Z_2}_\bullet(X)$.
  The homology associated to this chain complex is the \emph{equivariant homology} of $X$, denoted by $H_\bullet^{\Z_2}(X)$.

  Fix a $\grLambda$-module $N$, and consider the equivariant cochain complex:
  \[
    C_{\Z_2}^i (X; N) = \Hom_\grLambda\bigl(C^{\Z_2}_i(X), N\bigr)
  \]
  with the standard coboundary maps. The cohomology of this cochain complex is the \emph{Bredon cohomology}, denoted by $H_{\Z_2}^\bullet(X; N)$.
\end{definition}

We will use the following classical result to compute $[T^n, P]_\Ztwo$.

\begin{theorem}[{\cite[Theorem II.3.17]{Die87}; see also \cite[Chapter~II]{MayCP96}}]\label{lem:Eilenberg-MacLane}
  Let $\eml$ be a $\Ztwo$-CW complex which is an Eilenberg--MacLane space whose unique non-trivial homotopy group is $\pi_i(\eml)$ (we assume that $\pi_1(\eml)$ is abelian if $i=1$).
  Then, for every $\Ztwo$-CW complex $X$ such that there is a $\Ztwo$-map $X\to \eml$, the set $[X, \eml]_\Ztwo$ of $\Ztwo$-equivariant homotopy classes of equivariant maps is in bijection with $H^i_\Ztwo (X; \pi_i(\eml))$.
\end{theorem}

We will apply Theorem~\ref{lem:Eilenberg-MacLane} in the case where $X = T^n$ (with the diagonal action) and $\eml$ is the Eilenberg--MacLane space constructed in Lemma~\ref{lem:Postnikov}. The last remaining ingredient for the proof of Lemma~\ref{lem:number-homotopy-classes} is the following result on the Bredon cohomology of the torus $T^n$, which we will prove in Section~\ref{sec:computing_cohomology_torus} below:

\begin{proposition} \label{prop:cohomology_torus}
  For all $n, d\geq 1$,
  \[
    H_{\Z_2}^d\left(T^n; \pi_2(S^2)\right) \cong \Z_2^{\binom{n-1}{d-1}}.
  \]
\end{proposition}
\begin{proof}[Proof of Lemma~\ref{lem:number-homotopy-classes}]
By combining Theorem~\ref{lem:Eilenberg-MacLane} and Proposition~\ref{prop:cohomology_torus} and specializing to $d=2$, we get the following bijection:
\[
  [T^n,\eml]_\Ztwo \cong H^2_\Ztwo (T^n;\pi_2(S^2)) \cong \Z_2^{n-1}
\]
Thus, $  [T^n,\eml]_\Ztwo$ has $2^{n-1}$ elements, as we wanted to show.
\end{proof}  

\subsection{The equivariant cohomology of the torus}
\label{sec:computing_cohomology_torus}

The remainder of this section is devoted to proving Proposition~\ref{prop:cohomology_torus}. We begin with two technical lemmas that are useful for computing the equivariant cohomology of spaces with a free action.

\begin{lemma}
  \label{thm:free_module}
  If the action on $X$ is free and cellular, then $C_\bullet^{\Z_2}(X)$ is a chain complex of free modules over $\grLambda$.
\end{lemma}

\begin{proof}
  For every orbit of $d$-cells in $X$ choose a representative $\sigma$, and observe that the module $C_d(X)$ is freely generated by the set of these representatives.
\end{proof}

The above lemma implies that the functor $\Hom_\grLambda(C_d^{\Z_2}(X), {-})$ is exact for all free $\Z_2$-CW complexes $X$ and $d\geq 0$. Therefore, if we have a short exact sequence of $\grLambda$-modules
\[
  \begin{tikzcd}[cramped, sep=small]
    0\ar[r] & \frk N \ar[r, "g"] &\frk M \ar[r, "f"]&\frk Q\ar[r]&0
  \end{tikzcd}
\]
there is a corresponding short exact sequence of cochain complexes
\[
  \begin{tikzcd}[cramped, sep=small]
    0\ar[r] & C^\bullet_{\Z_2}\left(X;\frk N\right)\ar[r, "g_*"] & C_{\Z_2}^\bullet\left(X;\frk M\right) \ar[r, "f_*"]& C^\bullet_{\Z_2}\left(X;\frk Q\right)\ar[r]&0
  \end{tikzcd}
\]
and thus a long exact sequence in cohomology
\[
  \begin{tikzcd}[cramped, sep=small]
    \cdots \ar[r] & H_{\Z_2}^i\left(X; \frk N\right) \ar[r] & H_{\Z_2}^i(X;\frk M) \arrow[r] & H_{\Z_2}^i(X;\frk Q) \ar[r] & H^{i+1}_{\Z_2}(X; \frk M) \arrow[r] & \cdots
  \end{tikzcd}
\]

Note that $\Z$ admits exactly two non-isomorphic structures of $\grLambda$-module: either $\nu\cdot 1 = 1$ (in which case the action is trivial and we denote the module as $\zplus$) or $\nu\cdot 1 = -1$ (in which case the action is non trivial and we denote the module as $\zminus$).
We have the following diagram of $\grLambda$-modules, where $\frk I = (1+\gromega)\grLambda$ is the ideal generated by $1+\gromega$ and $\frk P$ is the module $\Z \oplus \Z$ where the action flips the two coordinates, i.e., $\gromega(n_0, n_1) = (n_1, n_0)$:
\[
\begin{tikzcd}
  \zplus \arrow[d, "d"', hook] \arrow[r, "\phi_1"] & \frk I \arrow[d, "\iota", hook] \\
  \frk P \arrow[r, "\phi_2"]          & {\grLambda}                  
\end{tikzcd}
\]
The maps in this diagram are defined as: $d(n) = (n, n)$, $\iota$ is the inclusion, $\phi_1(n) = n(1 + \gromega)$, and $\phi_2(n_0, n_1) = n_0 + n_1 \omega$. Moreover, note that $\phi_1$ and $\phi_2$ are isomorphisms of $\grLambda$-modules, hence the corresponding induced cochain maps
\begin{align*}
  (\phi_1)_*\colon &C_{\Z_2}^\bullet(X; \zplus)\to C_{\Z_2}^\bullet(X; \frk I)&\hfill
    & (\phi_1)_*\colon \alpha \mapsto \phi_1\circ \alpha \\
  (\phi_2)_*\colon &C_{\Z_2}^\bullet(X; \frk P) \to C_{\Z_2}^\bullet(X; \grLambda)&\hfill
    & (\phi_2)_*\colon \alpha \mapsto \phi_2\circ \alpha
\end{align*}
are isomorphisms of cochain complexes.

Assume now that $X$ is a CW complex with a free cellular $\Z_2$-action, and let $p\colon X\to {X} / {\Z_2}$ be the projection map that maps each element of $X$ to its orbit.
We have two isomorphisms of chain complexes of Abelian groups:
\begin{align*}
  h_1\colon C^\bullet(\faktor{X}{\Z_2}; \Z) &\to C^\bullet_{\Z_2}(X; \zplus)
    & h_1(\alpha)\colon \sigma &\mapsto \alpha(p(\sigma)) \\
  h_2\colon C^\bullet(X; \Z) &\to C^\bullet_{\Z_2}(X; \frk P)
    & h_2 (\alpha)\colon \sigma &\mapsto \bigl(\alpha(\sigma), \alpha(\gromega\cdot \sigma)\bigr).
\end{align*}

\begin{lemma} \label{lem:key_lemma_cohomology}
  Let $X$ be a CW complex with a free $\Z_2$-action, then the following diagram commutes
  \[
    \begin{tikzcd}
      C^\bullet(\faktor{X}{\Z_2}; \Z) \ar[d, "p^*"] \ar[r, "h_1"] & C^\bullet_{\Z_2}(X; \zplus) \ar[d, "d_*"]\ar[r, "(\phi_1)_*"] & C^\bullet_{\Z_2}(X; \frk I) \ar[d, "\iota_*"]\\
      C^\bullet(X; \Z) \ar[r, "h_2"] &C^\bullet_{\Z_2}(X; \frk P) \ar[r, "(\phi_2)_*"] & C^\bullet_{\Z_2}(X; \grLambda)
    \end{tikzcd}
  \]
  Moreover, as already noted, all the horizontal homomorphisms (i.e. $h_i$ and $(\phi_i)_*$) are isomorphisms of cochain complexes.
\end{lemma}

\begin{proof}
The right square commutes by functoriality of $C_{\Z_2}^\bullet(X; {-})$, therefore it is enough to show that the left square commutes, i.e., that, for any $d\geq 0$ and any $\alpha\in C^d(X/\Ztwo; \Z)$, $h_2(p^*(\alpha)) = d_*(h_1(\alpha))$:
We have, for all $\sigma\in C_d^\Ztwo(X)$,
  \begin{multline*}
  h_2(p^*(\alpha))(\sigma)
    = (p^*(\alpha)(\sigma), p^*(\alpha)(\nu\cdot\sigma))
    = \bigl(\alpha(p(\sigma)), \alpha(p(\nu\cdot\sigma))\bigr)
    = \bigl(\alpha(p(\sigma)), \alpha(p(\sigma))\bigr)
    \\
    = (h_1(\alpha)(\sigma), h_1(\alpha)(\sigma))
    = d_*(h_1(\alpha))(\sigma)
  \end{multline*}
as claimed.
\end{proof}

The next ingredient is the following well known result about the action of antipodality on $\pi_2(S^2)$.

\begin{lemma}
  The second homotopy group $\pi_2(S^2)$ is isomorphic as a $\grLambda$-module to $\zminus=\Z$ where the multiplication by $\gromega$ is $\gromega\cdot n = -n$.
\end{lemma}

\begin{proof}
  The statement follows from the fact that the degree of the antipodal map on $S^2$ is $-1$ which may be shown by observing that antipodality is homotopic to the reflection (see \cite[Section 2.2]{Hat02}).
  In more detail, since the action must be bijective, there are two options: either $\gromega \cdot n = n$, or $\gromega \cdot n = -n$. We can distinguish these two cases by evaluating on a generator of $\pi_2(S^2)$, i.e., on the homotopy class of $1_{S^2}$:
  \[
    \gromega\cdot [1_{S^2}] = [\gromega \circ 1_{S^2}] = [\gromega]
  \]
  Finally, the degree argument shows that $\gromega$ is not homotopic to the identity, and hence $\gromega\cdot n = -n$.
\end{proof}

The module $\zminus$ has some very useful properties that will allow us to use Lemma~\ref{lem:key_lemma_cohomology}. In particular, we have the following lemma.

\begin{lemma}\label{thm:coefficient_structure}
  The module $\zminus$ is generated by a single element and $\Ann(\zminus) = \frk I$ the ideal generated by $1+\gromega$.
\end{lemma}

\begin{proof}
  The module $\zminus$ is generated by $1$ therefore $\lambda \in \Ann(\zminus)$ if and only if $\lambda \cdot 1 = 0$; hence, if $\lambda = n_0 + n_1 \gromega $, then
  \[
    \lambda \cdot 1 = n_0 - n_1
  \]
  Therefore, $\lambda\cdot 1 = 0$ if and only if $n_0 = n_1$ if and only if $\lambda \in (1+\gromega)$.
\end{proof}

The last ingredient we will need is to determine what the projection map $p^*$ does on the level of cohomology.
While it is possible to compute $p^*$ directly, it is easier to view the action on the torus from a different perspective to simplify the calculations:

\begin{lemma}
  \label{thm:simple_torus}
  Let $X$ be the torus $T^n\subseteq \mathbb C^n$ with the diagonal $\Z_2$-action given by the multiplication with $-1$ (i.e., $\gromega\cdot(z_1, \dots, z_n) = (-z_1, \dots, - z_n))$.
  Let $Y$ be the same torus but with $\Z_2$ acting only on the first coordinate (i.e., $\gromega\cdot(z_1, \dots, z_n) = (- z_1, z_2, \dots, z_n)$).
  Then there is a $\Z_2$-equivariant homeomorphism $X\to Y$.
\end{lemma}

\begin{proof}
  The maps $h\colon X \to Y$ and $h'\colon Y \to X$ defined by
  \begin{align*}
    h\colon (z_1, \dots, z_n) &\mapsto (z_1, z_1^{-1} z_2, \dots, z_1^{-1} z_n) \\
    h'\colon (z_1, \dots, z_n) &\mapsto (z_1, z_1 z_2, \dots, z_1 z_n)
  \end{align*}
  are clearly continuous and mutually inverse.
  We will show that $h$ preserve the actions involved, and hence that $h$ is an equivariant homeomorphism:
  \[
    h(\gromega\cdot_X (z_1, \dots, z_n)) = h(- z_1, \dots, - z_n) = (- z_1, z_1^{-1} z_2, \dots, z_n z_1^{-1}) = \gromega\cdot_Y h(z_1, \dots, z_n)
   \]
\end{proof}

\begin{remark}
  If we view the torus $T^n$ as the quotient of $\R^n$ by the standard lattice $\Z^n$, then Lemma~\ref{thm:simple_torus} shows that factoring out the action is the same as factoring out the lattice generated by $\{\frac{1}{2} e_1, e_2, \dots, e_n\}$.
 Hence, topologically, the quotient is still a torus.
\end{remark}

Thus, for the remaining (co)homological calculations, we can assume that $\Z_2$ acts on $T^n$ by changing only on the first coordinate.
Using this simplified action on the torus it is much easier to compute the quotient map $p^*\colon H^\bullet(T^n/\Z_2)\to H^\bullet(T^n)$.
To achieve this objective, it is necessary to fix a basis for the cohomology of the torus. 
The ideal choice would be a basis that is ``easy'' to evaluate on homology classes in order to compute easily the image of $p^*$.

In the case of the torus, a direct application of the universal coefficient theorem \cite[Section 3.1]{Hat02} show that homology and cohomology in dimension $1$ are dual to each other; hence we can choose as basis for the first cohomology group the dual of a suitable basis for the first homology group.
In particular, let $\{x_i\}$ be the basis for $H_1(T^n)$ corresponding to the standard coordinate cycles in $C_1(T^n)$ (i.e., $x_i$ corresponds to the (non-equivariant) inclusion $S^1 \hookrightarrow T^n$, $z\mapsto (0,\dots,0,z,0,\dots,0)$ in the $i$th coordinate, $1\leq i\leq n$), and denote by $\{x^i\}$ the dual basis in $H^1(T^n)\cong \Hom\left(H_1(T^n), \Z\right)$; analogously, define bases $\{q_i\}$ of $H_1({T^n}/{\Z_2})$ and $\{q^i\}$ of $H^1({T^n}/{\Z_2})$.
Then
\[
  p^*_1(q^i) =
  \begin{cases}
    2x^1 & \text{if $i=1$}\\
    x^i & \text{otherwise.}
  \end{cases}
\]

The ring structure on cohomology (see \cite[Section 3.2]{Hat02}) of the torus allows us to build a convenient basis for all the other cohomology groups out of $\{x^i\}$.
In fact, elements of the form $x^I = x^{i_1}\smile \dots\smile x^{i_d}$, where $I = (i_1, \dots, i_d)$ and $i_1 < \dots < i_d$, form a basis for $H^d(T^n)$. Let $q^I$ denote the analogous basis for $H^d({T^n}/{\Z_2})$.
Since $p^*$ is a ring map, it commutes with the cup product, hence it can be explicitly computed on such a basis. We have that, for all $d$,
\[
  p^*_d(q^I) = p^*_1(q^{i_1})\smile \dots \smile p^*_1(q^{i_d})=
  \begin{cases}
    2x^I & \mbox{ if }i_1 = 1\\
     x^I & \mbox{ else.}
  \end{cases}
\]
In particular, $p^*_d$ is injective for all $d\geq 0$ and, in this choice of basis, $p^*$ is the diagonal matrix with $\binom{n-1}{d-1}$ 2's and $\binom {n-1}{d}$ 1's on the diagonal.

We are finally ready to compute the equivariant cohomology group of the torus $T^n$ and prove Proposition~\ref{prop:cohomology_torus}.

\begin{proof}[Proof of Proposition~\ref{prop:cohomology_torus}]
  Fix $n\geq 2$.
  By Lemma~\ref{thm:coefficient_structure}, we have a short exact sequence
  \[
    \begin{tikzcd}[cramped, sep=small]
      0\ar[r] & \frk I \ar[r] & \grLambda \ar[r]& \zminus\ar[r]&0
    \end{tikzcd}
  \]
  which induces short exact sequence of cochain complexes
  \[
    \begin{tikzcd}[cramped, sep=small]
      0\ar[r] & C^\bullet_{\Z_2}\left(T^n; I\right)\ar[r] & C_{\Z_2}^\bullet\left(T^n; \grLambda\right) \ar[r]& C^\bullet_{\Z_2}\left(T^n; \zminus\right)\ar[r]&0
    \end{tikzcd}
  \]
  Using Lemma~\ref{lem:key_lemma_cohomology}, we get that the following short sequence is also exact
  \[
  \begin{tikzcd}[cramped, sep=small]
    0\ar[r] & C^\bullet\left(\faktor{T^n}{\Z_2}\right)\ar[r, "p^*"] & C^\bullet\left(T^n\right) \ar[r]& C^\bullet_{\Z_2}\left(T^n; \zminus\right)\ar[r]&0
  \end{tikzcd}
  \]
  This short exact sequence induces the following long exact sequence in cohomology
  \[
    \begin{tikzcd}[cramped, sep=small]
      \cdots \ar[r]
      & H^d\left(\faktor{T^n}{\Z_2}\right) \ar[r, "p^*_d"]
      & H^d\left(T^n\right) \ar[r]
      & H_{\Z_2}^d\left(T^n; \zminus\right)
      \arrow[r] & H^{d+1}\left(\faktor{T^n}{\Z_2}\right) \ar[r, "p^*_{d+1}"] & H^{d+1}\left(T^n\right) \ar[r] & \cdots
    \end{tikzcd}
  \]
Since $p^*_d$ is injective for any $d\geq 1$, by exactness we have that
  $H_{\Z_2}^d\left(T^n; \zminus\right) \cong \coker p^*_d$.
  Finally,
  \[
    \coker p^*_d = \Z^{\binom nd} / \operatorname{im} p^*_d \simeq \Z_2^{\binom{n-1}{d-1}}
  \]
  which yields the desired result.
\end{proof}

\section{Topological hardness conjecture}

In this paper, we focused on promise homomorphism problems of graphs, and in particular on spaces $\Hom(K_2, G)$, which admit a natural $\mathbb Z_2$ action. Nevertheless, the topological method was also successfully demonstrated for other relational structures and other actions; see, e.g., \cite{FilakovskyNOTW26}.
In this section, we suggest a generalisation of the Brakensiek--Guruswami conjecture to a wider class of promise CSPs.

Recall that a \emph{relational structure} $\rel A$ consists of a domain $A$, and a sequence of relations $R^\rel A, S^\rel A, \dots$, each of which is a subset of $A^k$ for some $k$ which depends on the relational symbols $R, S, \dots$, but not on the structure. A \emph{homomorphism} from $\rel A$ to $\rel B$ is a mapping $f\colon A \to B$ which preserves the relations in the sense that $(f(a_1), \dots, f(a_k)) \in R^\rel B$ if $(a_1, \dots, a_k) \in R^\rel A$.
Most of the topological methods recently used in CSPs apply equally to arbitrary relational structures in place of graphs, and are usually formulated for graphs only to increase accessibility. In particularly, multihomomorphisms and homomorphism complexes are defined analogously to our definitions in Paragraph~\ref{par:homomorphism-complexes}; see also~\cite[Definitions 2.7 and 2.8]{FilakovskyNOTW26}.

Using these notions, we conjecture the following topological condition on the hardness of promise CSPs.

\begin{conjecture}
  Let $\rel A, \rel B$ be two finite relational structures such that there is a homomorphism $f\colon \rel A \to \rel B$. Furthermore, assume that there exists a structure $\rel T$ with an automorphism $\alpha$ of prime order $p$ such that:
  \begin{enumerate}
    \item $\Hom(\rel T, \rel A)$ is connected, and
    \item the induced action of $\alpha$ on $\Hom(\rel T, \rel B)$ is fixed-point free.
  \end{enumerate}
  Then $\PCSP(\rel A, \rel B)$ is NP-complete.
\end{conjecture}

Equivalently, the conjecture could be phrased as: There are equivariant maps
\[
  S^1 \to \Hom(\rel T, \rel A) \to \Hom(\rel T, \rel B) \to S^{2d+1}
\]
for some $d < \infty$. Where we consider fixed-point free actions on the odd-dimensional spheres.\footnote{Note that odd dimensional spheres $S^{2d+1}$ admit a free $\mathbb Z_p$ action, unlike even dimensional spheres.}

Brakensiek--Guruswami conjecture immediately follows from the case $p = 2$, and this has already been noted in \cite[Section 5]{KOWZ23}. Similarly, the methods of \cite{FilakovskyNOTW26} show that the above conjecture implies NP-hardness of colouring 3-uniform hypergraphs, that admit a \emph{unique maximum 3-colouring}, with any constant number of colours.
We support the conjecture with three pieces of evidence.

Firstly, in the light of \cite[Theorem 4.32]{KOWZ23}, which states that $\PCSP(C_{2\ell+1}, H)$ only depends on the topology of $\Hom(K_2, H)$ as a $\mathbb Z_2$-space, it is believable that the case $p = 2$ might be in fact equivalent to the Brakensiek--Guruswami conjecture as stated in the following proposition.

\begin{proposition} \label{prop:p-is-2}
  The case $p = 2$ is equivalent to the Brakensiek--Guruswami conjecture.
\end{proposition}

\begin{proof}
  The fact that Brakensiek--Guruswami conjecture follows from our conjecture for $p = 2$ should be clear at this point. We focus on proving the converse.

  The core of the proof is \emph{pp-constructing} \cite[Section 4.2]{BBKO21} a pair of graphs $G$ and $H$, which is in the scope of Brakensiek--Guruswami conjecture, from the template $\rel A$, $\rel B$ satisfying the assumption of our conjecture.
  On the high-level, one such construction would be to observe that Csorba's construction of a graph from a $\mathbb Z_2$-complex \cite{Csorba2007} (see also \cite{Matsushita2017}) is in fact expressible as a pp-construction. The resulting graph is then pp-constructible from the template since the homomorphism complex is.
  To avoid doubt that this can be done, we present a direct combinatorial argument that essentially unfolds these constructions.

  We use the structural rather then logical perspective on pp-constructions (see, e.g., \cite[Appendix A.1]{HadekJO2026-arXiv} for a comparison of the logical and structural perspectives), and we will only use a restricted notion of pp-constructions called \emph{pp-powers} \cite[Definition 4.7]{BBKO21}, i.e., the vertex set of the constructed (di)graph is a power of the domain of $A$, i.e., $A^k$. The edges are then defined using a structure $\rel E$ with $2k$ distinguished points $x_1, \dots, x_k, y_1, \dots, y_k$:
  A pair of $k$-tuples $(a, b)$ is an edge if there is a homomorphism $h \colon \rel E \to \rel A$ such that $h(x_i) = a_i$ and $h(y_i) = b_i$.

  For the rest of the proof, assume that $\rel A$, $\rel B$, and $\rel T$ satisfy the assumptions of the conjecture for $p = 2$.
  \begin{enumerate}
    \item
      We use $\rel T$ to pp-construct a symmetric binary relation defining the edges of $G$ and $H$. In order to do so, we need to pick the distinguished points $x_1, \dots, x_k, y_1, \dots, y_k$:
      Consider the decomposition of the action of $\alpha$ on vertices of $\rel T$ into cycles; these are either transpositions or fixed-points. Pick $x_i$'s and $y_i$'s in the same transposition, so that the action is written as
      \[
        (x_1, y_1)(x_2, y_2)\cdots(x_k, y_k).
      \]
      Consider the digraphs pp-constructed using this $\rel T$ from $\rel A$ and $\rel B$, let us denote them by $G$ and $H$, respectively. Observe that the binary relation defined by $\rel T$ on $A^k$ and $B^k$, resp., is symmetric due to the action of $\alpha$. Hence, both $G$ and $H$ are in fact unoriented graphs.
    \item
      Next, we show that $G$ is not bipartite using the fact that $\Hom(\rel T, \rel A)$ is connected. We use Lovász's method \cite[Theorem 2]{Lovasz-Kneser}, although for 2-colouring it might seem to be unnecessarily powerful, and prove that by showing that $\Hom(K_2, G)$ is connected.
      Observe that there is a map $\hom(\rel T, \rel A) \to \hom(K_2, G)$ between the two homomorphism sets defined by composition of homomorphisms, which extends to a monotone map $h \colon \mhom(\rel T, \rel A) \to \mhom(K_2, G)$ defined as
      \begin{align*}
        h(m) \colon
        0 &\mapsto m(x_1) \times \dots \times m(x_k), \\
        1 &\mapsto m(y_1) \times \dots \times m(y_k).
      \end{align*}
      It is straightforward to check that this map is indeed well-defined, monotone, and the image contains all homomorphisms.
      Consequently, it induces a simplicial (in fact $\mathbb Z_2$-equivariant) map
      \[
        \Hom(\rel T, \rel A) \to \Hom(K_2, G).
      \]
      Since the domain is connected, the image is connected. Furthermore, since the image hits every connected component (since each such component contains a homomorphism), this implies that $\Hom(K_2, G)$ is connected, and hence not 2-colourable by Lovász's theorem.
    \item
      Finally, we show that $H$ is loopless. This follows from the fact that the action of $\alpha$ on $\Hom(\rel T, \rel B)$ is fixed-point free: A loop in $H$ would be witnessed by a homomorphism $h \colon \rel T \to \rel B$ such that $h(x_i) = h(y_i)$ for all $i$, but such a homomorphism would be fixed by $\alpha$.
  \end{enumerate}
  The hardness of $\PCSP(\rel A, \rel B)$ follows from the hardness of $\PCSP(G, H)$ by \cite[Corollary 9.5]{BBKO21}.
\end{proof}

Secondly, we show that no template $\rel A$, $\rel B$ within the scope of the conjecture admits a \emph{finite tractable sandwich}, i.e., for every finite structure $\rel C$ such that $\rel A \to \rel C \to \rel B$, $\CSP(\rel C)$ is NP-complete, this is based on topological characterization of the CSP dichotomy by Meyer~\cite{Meyer24}, and in particular, the following theorem (which is a combination of Corollary~3.1 and Lemma~4.1 of \cite{MeyerO24}).

\begin{theorem}[Meyer~\cite{Meyer24}]
  Let $\rel A$ and $\rel T$ be two relational structures of the same signature, such that one of the connected components of $\Hom(\rel T, \rel A)$ is not contractible, then $\CSP(\rel A)$ is NP-complete.
\end{theorem}

\begin{corollary}
  For any pair of structures $\rel A$ and $\rel B$ within the scope of the conjecture, every finite sandwich $\rel C$ has NP-complete CSP.
\end{corollary}

\begin{proof}
  Let $\rel A \to \rel C \to \rel B$. Using the above theorem, we only need to show that $\Hom(\rel T, \rel C)$ has a non-contractible connected component.
  Observe that there are equivariant maps
  \[
    S^1
      \to \Hom(\rel T, \rel A)
      \to \Hom(\rel T, \rel C)
      \to \Hom(\rel T, \rel B)
      \to S^{2\ell + 1}.
  \]
  Consequently, the connected component containing the image of $S^1 \to \Hom(\rel T, \rel C)$ cannot be contractible as otherwise the action would have a fixed point by the Lefschetz fixed-point theorem \cite[Theorem 2C.3]{Hat02}, which is in turn ruled out by the equivariant map $\Hom(\rel T, \rel C) \to S^{2\ell + 1}$.
\end{proof}

Thirdly, Banakh, Barsukov, and Nakajima~\cite[Theorem 8]{BanakhBN26} recently showed that a variation of Brakensiek--Guruswami conjecture for structures with a single (cyclically) symmetric $p$-ary relation follows from the \emph{Smooth 2-to-1 Conjecture}. With a similar argument as in Proposition~\ref{prop:p-is-2} but using an arbitrary prime $p$ in place of $p = 2$, it could be also shown that our conjecture follows from the Smooth 2-to-1 Conjecture.

Apart from just being a generalisation of the Brakensiek--Guruswami conjecture, resolving our conjecture would provide a uniform explanation of hardness of approximate graph and hypergraph colouring (see \cite{FeiMW26,GS20} and \cite{DRS05}, resp.).
Several well-studied promise CSPs are also in the scope of the conjecture. Most prominently \emph{approximate unique maximum hypergraph colouring}; in this variant of colouring (3-uniform) hypergraphs, the colours come with a linear order on them, and the goal is to find a colouring such that, for each edge, the maximal colour appears exactly once in the edge. Colouring 3-uniform hypegraphs with 2 colours therefore corresponds to the well-known (monotone) 1-in-3SAT problem. We refer to Krokhin and Vagnozzi \cite{KrokhinV26} for an excellent introduction, and overview of these problems.
The conjecture immediately implies hardness of colouring a 3-uniform hypergraph, which is promised to be unique-maximum colourable with 3 colours, with any finite number of colours in the usual way, i.e., NP-hardness of the problem which is often denoted by $\PCSP(\rel{LO}_3, \rel{NAE}_\ell)$ for all $\ell \geq 3$. In particular, it implies the Barto--Battistelli--Berg conjecture \cite[Conjecture 27]{BBB21} for $\ell \geq k \geq 3$.
The promise version of 1-in-3SAT itself is more subtle since, for any structure $\rel T$, the homomorphism complex $\Hom(\rel T, \rel{LO}_2)$ is discrete. The hardness line also lies elsewhere; see \emph{Approximate 1-in-3SAT Conjecture} \cite[Conjecture 5]{NakajimaZ24} (the name was coined in \cite{KrokhinV26}).

Importantly, the criterion given in our conjecture is ot stable under algebraic reductions between PCSPs. In particular, two promise templates might have the same polymorphisms, and hence the same complexity, while the criterion applies to one, but not the other. As far as we know, this has not yet been sufficiently explored.

\bibliographystyle{ACM-Reference-Format}

\end{document}